\documentclass[
pra,aps,
reprint,
a4paper,
superscriptaddress,
longbibliography,
floatfix
]{revtex4-1}
\usepackage{times}
\usepackage{epsfig}
\usepackage{amsfonts}
\usepackage{amsmath}
\usepackage{amssymb}
\usepackage{amsthm}
\usepackage{color}
\usepackage{bbm}
\usepackage[normalem]{ulem}
\usepackage{latexsym}
\usepackage{amsfonts}
\usepackage{algorithmicx}
\usepackage{algpseudocode}
\usepackage{mathrsfs}
\usepackage{natbib}
\usepackage{verbatim}
\usepackage{multirow} 
\usepackage{multirow}

\usepackage{times}
\usepackage{epsfig}
\usepackage{amsfonts}
\usepackage{amsmath}
\usepackage{amssymb}
\usepackage{amsthm}
\usepackage{color}
\usepackage{bbm}
\usepackage[normalem]{ulem}
\usepackage{latexsym}
\usepackage{amsfonts}
\usepackage{algorithm}
\usepackage{algorithmicx}
\usepackage{algpseudocode}
\usepackage{amsthm}
\usepackage{mathrsfs}
\usepackage{natbib}
\usepackage{enumerate}
\usepackage{verbatim}
\usepackage{multirow} 
\usepackage{dsfont}
\usepackage{gensymb}

\usepackage[table,xcdraw]{xcolor}
\newtheorem{proposition}{Proposition}
\newcommand{\tr}{\ensuremath{\operatorname{tr}}}
\DeclareMathOperator{\Span}{span}

\usepackage[colorlinks=true,linkcolor=blue,citecolor=magenta,urlcolor=blue]{hyperref}

\DeclareMathOperator{\Tr}{tr}
\newtheorem{theorem}{Theorem}[section]

\newcommand{\ket}[1]{|#1\rangle}

\newcommand{\braket}[2]{\langle#1|#2\rangle}

\newcommand{\ketbra}[2]{|#1\rangle\langle#2|}

\newcommand{\id}{\mathbbm{1}}

\begin{document}


\title{Enabling computation of correlation bounds for finite-dimensional quantum systems via symmetrisation}

\author{Armin Tavakoli} \thanks{A. T. and D. R. contributed equally for this project.}
\affiliation{Department of Applied Physics, University of Geneva, 1211 Geneva, Switzerland}

\author{Denis Rosset} \thanks{A. T. and D. R. contributed equally for this project.}
\affiliation{Perimeter Institute for Theoretical Physics, Waterloo, Ontario, Canada, N2L 2Y5}

\author{Marc-Olivier Renou}
\affiliation{Department of Applied Physics, University of Geneva, 1211 Geneva, Switzerland}

\date{\today}

\begin{abstract}
  We present a technique for reducing the computational requirements by several orders of magnitude in the evaluation of semidefinite relaxations for bounding the set of quantum correlations arising from finite-dimensional Hilbert spaces.  The technique, which we make publicly available through a user-friendly software package, relies on the exploitation of symmetries present in the optimisation problem to reduce the number of variables and the block sizes in semidefinite relaxations. It is widely applicable in problems encountered in quantum information theory and enables computations that were previously too demanding. We demonstrate its advantages and general applicability in several physical problems. In particular, we use it to robustly certify the non-projectiveness of high-dimensional measurements in a black-box scenario based on self-tests of $d$-dimensional symmetric informationally complete POVMs.
\end{abstract}


\maketitle


\textit{Introduction.---}  Finite-dimensional quantum systems are common in quantum information theory. They are standard in the broad scope of quantum communication complexity problems (CCPs) \cite{qCCP} in which quantum correlations are studied under limited communication resources. Furthermore, they are widely used in semi-device-independent quantum information protocols \cite{sdi1} in which systems are fully uncharacterised up to their Hilbert space dimension. Also, studying correlations obtainable from finite-dimensional systems is critical for device-independent dimension witnessing \cite{DW, DW2}.

In view of their diverse relevance, it is important to bound quantum correlations arising from dimension-bounded Hilbert spaces. To this end, semidefinite programs (SDPs) \cite{SDP} constitute a powerful tool. Lower bounds on quantum correlations are straightforwardly obtained using alternating convex searchers (SDPs in see-saw) \cite{seesaw2, seesaw}. However, obtaining upper bounds valid for \textit{any} quantum states and measurements is more demanding. A powerful approach to this problem is to relax some well-chosen constraints of quantum theory so that the resulting super-quantum correlations easily can be computed with SDPs, thus returning upper bounds on quantum correlations. Such approaches are commonplace in various problems in quantum information theory \cite{NPA, Moroder, NV}. A hierarchy of semidefinite relaxations for upper-bounding quantum correlations on  dimension-bounded Hilbert spaces was introduced by Navascu\'es and V\'ertesi (NV) \cite{NV, NV2}. This is an effective tool for problems involving a small number of states and measurements, and low Hilbert space dimensions. However beyond simple scenarios, the computational requirements of evaluating the relaxations quickly become too demanding. 

It is increasingly relevant to overcome the practical limitations of the NV hierarchy, i.e. to provide efficient computational tools for  bounding quantum correlations in problems beyond small sizes and low Hilbert space dimensions. This is motivated by both theoretical and experimental advances. Dimension witnessing has been experimentally realised far beyond the lowest Hilbert space dimensions \cite{AB14, AF18}. Furthermore, increasing the dimension can activate unexpectedly strong quantum correlations \cite{magic7}; a phenomenon that has been experimentally demonstrated \cite{DM18}. Also, quantum correlations obtained from a sizeable number of states and measurements are interesting for studying mutually unbiased bases \cite{Sym8}.  Moreover, large problem sizes naturally appear in multipartite CCPs involving single particles \cite{Galvao, TS05, ST16}. Similarly sized problems also appear in multipartite CCPs for the characterisation of entangled states and measurements \cite{TA18}. In addition, efficiently evaluating the NV hierarchy many times can improve randomness extraction from experimental data \cite{MT16}.

In this work we develop techniques for efficiently bounding quantum correlations under dimension constraints. The technique is powered by the exploitation of \textit{symmetries}, i.e. re-labellings of optimisation variables that leave a figure of merit invariant. The use of symmetries for reducing the complexity of SDPs was first introduced in \cite{Parrilo} and was shown to lead to remarkable efficiency gains. These efficiency gains have also been harvested in several specific quantum information problems relying on SDPs. These include finding bounds on classical  \cite{Sym3} and quantum \cite{Sym7, Sym4} Bell correlations, quantifying entanglement \cite{Moroder, BancalAdditional}, and finding symmetric Bell inequalities \cite{SymJD}. Note that symmetries in Bell scenarios also have been studied without application to SDPs \cite{Sliwa, CG, 50years, Renou2016}. In dimension-bounded scenarios, symmetries have been considered for CCPs tailored for studying the existence of mutually unbiased bases \cite{Sym8}.

We describe a powerful, generally applicable, and easy-to-use technique for symmetrised semidefinite relaxations for dimension-bounded quantum correlations. We show how to automatise searches for symmetries in general Bell scenarios and CCPs, and how these can be exploited to reduce computational requirements in all parts of the NV hierarchy. This amounts to reducing the number of variables in an optimisation, and reducing block sizes beyond previous approaches. We make these techniques readily available via a user-friendly software package supporting general correlation scenarios. Subsequently, we give examples of problems that can be solved faster (several orders of magnitude), and other previously unattainable problems that can now be computed. We focus on the usefulness of symmetrisation for the problem of certifying that an uncharacterised device implements a non-projective measurement using only the observed correlations. To this end, we introduce a family of CCPs, prove that they enable self-tests of $d$-dimensional symmetric informationally complete (SIC) POVMs, then use symmetrised semidefinite relaxations to bound the correlations attainable under projective measurements. This allows us to go beyond previously studied qubit systems \cite{APV16, GGG16, Armin, Piotr, Massi} and robustly certify the non-projectiveness of SIC-POVMs subject to imperfections.

\textit{Bounding finite-dimensional quantum correlations.---}
We begin by summarising the NV hierarchy \cite{NV, NV2} for optimising dimensionally constrained quantum correlations.
For simplicity, we first describe CCPs, and later consider Bell scenarios.

Consider a CCP in which a party, Alice, holds a random input $x$ and another party, Bob, holds a random input $y$.
Alice encodes her input into a quantum state $\rho_x$ of dimension $d$ and sends it to Bob.
Bob performs a measurement $\{M_y^b\}_b$ with outcome $b$.
The resulting probability distribution is used to evaluate a functional $F(P)=\sum_{x,y,b}c_{x,y}^bP(b|x,y)$, where $c_{x,y}^b$ are real coefficients.
The problem of interest is to compute the maximal quantum value of $F$ when the probabilities are given by the Born rule $P(b|x,y)=\Tr\left(\rho_xM_y^b\right)$, where the measurement operators are taken to be projectors.
The NV hierarchy presents the following semidefinite relaxations.
Sample a random set of states and measurements $\{\rho_x\}$ and $\{M_y^b\}$ of dimension $d$, which we collect in the set of operator variables $\{X_i\}$. Then, generate all strings, $\{s_j(X)\}_j$, of products of at most $L$ of these operators.
The choice of $L$ determines the degree of relaxation, i.e., the level of the hierarchy. Construct a moment matrix
\begin{equation}
  \label{moment}
  \Gamma_{j,k}=\Big < s_j(X)^\dagger~ s_k(X) \Big >\;,
\end{equation}
where, for the present CCP, the expectation value of an operator product $S$ is $\left< S \right> = \Tr S$.
Repeat this process many times, each time obtaining a new moment matrix.
Terminate the process when the sampled moment matrix is linearly dependent on the collection of those previously generated.
Hence, $\{\Gamma^{(1)},\ldots,\Gamma^{(m)}\}$ identifies a basis for the feasible affine subspace $\mathcal{F}$ of such matrices under the given dimensional constraint. The semidefinite relaxation amounts to finding an affine  combination $\Gamma=\sum_{\ell=1}^m c_\ell\Gamma^{(\ell)} \in \mathcal{F}$, with $\Gamma\geq 0$, that maximises the functional $F$ (which can be expressed as a linear combination of entries of $\Gamma$).
Hence, the relaxation reads
\begin{align}\label{sdp}
  \max_{\vec{c}\in\mathbb{R}^m} ~  F(\Gamma) ~~~~~~~~  \text{s.t. } ~~~~~~~~\Gamma  \geq 0, ~~ \sum_{\ell=1}^m c_\ell  = 1.
\end{align}

In summary, the problem consists in first sampling a basis enforcing the dimensional constraint and then evaluating an SDP.
Crucially, the complexity of solving the SDP hinges on the number of basis elements, $m$, needed to complete the basis and the size of the final SDP matrix, $n$. For a single iteration of primal-dual interior point solvers, the required memory scales as $\mathcal{O}(m^2+mn^2)$ while the CPU time scales as $\mathcal{O}(m^3+n^3+mn^3+m^2n^2)$~\cite{Personal1}.
Without exploitation of the problem structure, medium-sized physical scenarios, as well as small-sized scenarios with high relaxation degree, practically remain out of reach for current desktop computers.

\textit{Symmetric relaxations.---} The key to reducing the computational requirements for the NV hierarchy is two-fold; first reducing the number of elements needed to form the basis in the sampling step, i.e.,  decreasing the dimension of $\mathcal{F}$ and then shrinking the size of the positivity constraints in the subsequent SDP by block-diagonalising $\Gamma$. Here, we show how such a reduction can be systematically achieved by identifying and exploiting the set of symmetries of the problem.

Recall that $\{X_i\}$ collects all the operators (states, measurements etc.) present in the formulation of the problem, where $i\in\mathcal{I}$ is an index. Consider a permutation of elements of $\mathcal{I}$, i.e., a bijective function $\pi: \mathcal{I} \to \mathcal{I}$.
We write $\pi(X_i) = X_{\pi(i)}$ and define the action of the permutation on the strings $s=X_iX_j\ldots$ of products of operators $X_i$ appearing in the NV hierarchy as $\pi(X_i X_j \ldots) = X_{\pi(i)} X_{\pi(j)}\ldots$.
We call $\pi$ an {\em ambient symmetry} if it is a transformation of the scenario which preserves its structure, as expressed by implicit or explicit constraints on the operators $\{X_i\}$.
The set of those symmetries form the {\em ambient group} $\mathcal{A} = \{ \pi \}$. In Supplementary Material (SM), we describe the ambient groups for general Bell scenarios and CCPs. Given a moment matrix $\Gamma$ and $\pi\in \mathcal{A}$, we consider the re-labelled matrix $\pi(\Gamma)$ where $\big(\pi(\Gamma)\big)_{j,k} = \Gamma_{\pi^{-1}(j),\pi^{-1}(k)}$, according to the convention of Eq.~\eqref{moment}.
By construction, $\pi$ preserves the constraints of the problem: for a feasible moment matrix $\Gamma\in\mathcal{F}$ we have $\pi(\Gamma)\in\mathcal{F}$ for any $\pi\in \mathcal{A}$.
Moreover, the feasible set $\mathcal{F}$ is convex, so any convex combination of those $\pi(\Gamma)$ is feasible as well.

However, not all elements of $\mathcal{A}$ leave the objective $F(\Gamma)$ invariant.
We write $\mathcal{G} = \{ \pi \in \mathcal{A} : F(\pi(\Gamma)) = F(\Gamma) \}$ the {\em symmetry group} of the optimisation problem.
One can straightforwardly find the elements of $\mathcal{G}$ by enumerating the elements of $\mathcal{A}$ and filtering those that leave $F(\pi(\Gamma)) = F(\Gamma)$ invariant. Then, following a standard procedure~\cite{Parrilo,SymJD,Sym7, Sym8} we can average any optimal solution $\Gamma$ under the Reynolds operator, defined as:
\begin{equation}\label{Reynolds}
\Gamma' \equiv \mathcal{R}(\Gamma) = \frac{1}{|\mathcal{G}|}\sum_{\pi\in \mathcal{G}} \pi(\Gamma)
\end{equation}
where $|\mathcal{G}|$ is the size of $\mathcal{G}$ and obtain an optimal solution of the problem, which now satisfies $\pi(\Gamma') = \Gamma'$ for all $\pi\in \mathcal{G}$.  Since the set $\Gamma'$ is characterised by the relation $\mathcal{R}(\Gamma')=\Gamma'$, instead of searching the optimal $\Gamma$ in the full feasible set, it is sufficient to only consider the symmetric subspace $\mathcal{R}(\mathcal{F})$ given by the image of the feasible set under $\mathcal{R}$.
As discussed above, the basis of $\mathcal{F}$ is found by sampling.
To sample $\mathcal{R}(\mathcal{F})$ instead, we simply apply $\mathcal{R}$ on each sample during the construction of the basis, thus obtaining $\{{\Gamma'}^{(1)}, \ldots, {\Gamma'}^{(m')}\}$.
As a result, the size of the basis, $m'$, decreases due to the smaller dimension of $\mathcal{R}(\mathcal{F})$. In SM, we discuss methods for speeding up the computation of $\mathcal{R}$.

Moreover, a second major reduction is obtained: as the symmetrised moment matrices $\Gamma'$ commute with a representation of the group $\mathcal{G}$, there exists~\cite{Parrilo} a unitary matrix that block-diagonalises the moment matrix. This reduces the size of the positivity constraint on the final SDP matrix. A complete symmetry exploitation is obtained when the decomposition of the representation of $\mathcal{G}$ into irreducible components with multiplicities is known. We achieve this via an efficient general block-diagonalisation method detailed in SM. Moreover, we make available a user-friendly MATLAB package \cite{package} for symmetrisation of semidefinite relaxations in the NV hierarchy applicable to general correlation scenarios encountered in quantum information. The package automates both a search for the symmetries of a problem (if these are unknown) and the construction of symmetry-adapted relaxation.

\textit{Robust  certification of non-projective measurements based on SIC-POVMs.---} We now exemplify the usefulness of symmetrisation in a physical application.  We certify, solely from observed data, that an uncharacterised device ('black-box') implements a non-projective measurement. Non-projective measurements have diverse applications in quantum theory  \cite{USD1, USD2, DB98, Renes, Shang, APV16, GM17, Brask}. This has motivated interest in their black-box certification \cite{APV16, GGG16, Armin, Piotr, Massi}. Using semidefinite relaxations (whose complexity scale quickly with dimension) as a primary tool, these works limit themselves to qubits.   We use symmetrisation to overcome this limitation and certify the non-projectiveness of higher-dimensional measurements of physical interest. The latter is of particular importance; a certificate is typically only useful for non-projective measurements that are close (e.g. in fidelity) to   a particular targeted non-projective measurement  corresponding to the optimal quantum correlations \cite{Armin}. 

One of the most celebrated non-projective measurements are SIC-POVMs. These are sets of $d^2$ sub-normalised rank-one projectors $\{\frac{1}{d}\ketbra{\psi_x}{\psi_x}\}_{x=1}^{d^2}$ with  $\lvert \braket{\psi_x}{\psi_{x'}} \rvert^2=1/(d+1)$ when $x\neq x'$. Higher-dimensional SIC-POVMs have been of substantial interest for both fundamental (see e.g.~\cite{Fuchs} for a review) and practical considerations \cite{SIC1, SIC2, SIC3, SIC4, SIC5} in quantum information theory. We introduce a family of CCPs and prove that optimal quantum correlations imply a $d$-dimensional SIC-POVM. However, due to unavoidable experimental imperfections, such optimal correlations will never occur in practice. Therefore, we use symmetrisation to certify the non-projectiveness of measurements close to SIC-POVMs, that achieve nearly-optimal correlations. Moreover, as noted in \cite{Armin}, the  dimension-bounded scenario is well-suited for black-box studies of non-projective measurements since said property is only well-defined on Hilbert spaces of fixed dimension.

Consider a CCP in which Alice encodes her input $x$ into a $d$-dimensional system sent to Bob, who associates his input $y$ to a measurement producing an outcome $b$.
\begin{figure}
	\centering
	\includegraphics[width=0.95\columnwidth]{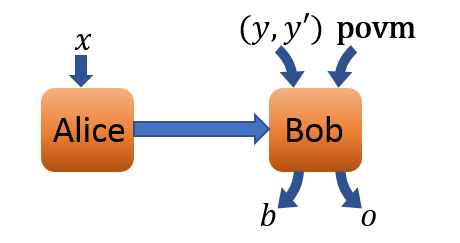}
	\caption{Illustration of the CCP \eqref{witness}. Bob has $\binom{N}{2}$ settings labelled by $(y,y')$ and one additional setting labelled $\mathbf{povm}$. Alice and Bob aim to satisfy the following relations: $o=x$ for the setting $\mathbf{povm}$, and $b=0$ when $x=y$ and $b=1$ when $x=y'$ respectively for the settings $(y,y')$.}\label{fig1}
\end{figure}
A general witness can be written
\begin{equation}
W=\sum_{x,y,b} \alpha_{xyb}P(b|x,y),
\end{equation}
where $\alpha_{xyb}$ are real coefficients. By tuning the coefficients, one can construct CCPs in which the optimal correlations $W^Q$ are uniquely realised with a particular non-projective measurement. This is known as a self-test \cite{TK18}. Consequently, there must exist some  $W^P<W^Q$ which bounds the correlations under all projective measurements. Thus, observing $W>W^P$ certifies that Bob implements a non-projective measurement.
 
We construct a family of CCPs (inspired by Refs~\cite{Armin, BNV13}) tailored to self-test $d$-dimensional SIC-POVMs.  Alice and Bob each receive inputs  $x\in[N]$ and $(y,y')\in[N]$ with $y<y'$ respectively, for some $N>d$ and $[N]=\{1,\ldots,N\}$. Bob outputs $b\in\{0,1\}$. Bob also possesses another  measurement setting labelled $\mathbf{povm}$ which returns an outcome $o\in [N]$. The witness of interest is 
\begin{multline}\label{witness}
W_{d}=\sum_{x<x'} P(b=0|x,(x,x'))+P(b=1|x',(x,x'))\\
+\sum_{x=1}^{N}P(o=x\lvert x,\mathbf{povm}),
\end{multline} 
The scenario is illustrated in Figure \ref{fig1}. 

\begin{theorem} For $N=d^2$, the maximal quantum value of the witness is
	\begin{equation}\label{Qmax}
	W^Q_{d}=\frac{1}{2}\sqrt{d^5(d-1)^2(d+1)}+\binom{d^2}{2}+d.
	\end{equation}
	This value self-tests that Alice prepares a SIC-ensemble and that Bob's setting $\mathbf{povm}$ corresponds to a SIC-POVM. 
\end{theorem}
The proof is given in SM. To enable the certification of a non-projective measurement producing nearly-optimal correlations, we must obtain a bound $W^P_d$ on $W_d$ respected by all projective measurements. To this end, we use symmetrised semidefinite relaxations. 

The symmetries of the witness \eqref{witness} correspond to coordinated permutations of the inputs of Alice and inputs and outputs of Bob. We permute $x$ among its $N$ possible values. This requires us to compensate the permutation by also applying it to $o$. Furthermore, to preserve the probabilities appearing in the first summand of $\eqref{witness}$, we must apply a permutation to the indices $(y,y')$ and the outcome $b$. Moreover, since we are interested in bounding $W_d$ under projective measurements, said property must be explicitly imposed on Bob's setting $\mathbf{povm}$. This means that at most $d$ of the POVM elements $\{M_\mathbf{povm}^x\}_{x=1}^{d^2}$ are non-zero, corresponding to  rank-one projectors. This must be accounted for in the symmetries of the problem. In SM we discuss the symmetries in  detail.

\begin{table}
	\centering
	\begin{tabular}{|c|c|c|c|c|c|}
		\hline
		$d$ & 2  & 3  &  4 & 5 & 6\\ [0.5ex]
		\hline
		LB: $W_d^P$ &  12.8484  & 70.0961 &  231.2685 & 578.7002 & 1219.0129\\
		UB: $W_d^P$ &  12.8484 &  70.1133 & 231.2685 & 578.7987 & 1219.2041  \\
		 $W_d^Q$ & 12.8990  &  70.1769 & 231.3313 & 578.8613 & 1219.2667\\				
		\hline
	\end{tabular}
	\caption{Upper bounds (UB) and lower bounds (LB) on quantum correlations under projective measurements with $N=d^2$. The lower bounds are obtained via SDPs in alternate convex search and the upper bounds via symmetrised semidefinite relaxations.}\label{tablenoproj}
\end{table}

Using the general recipe, we have implemented the symmetrised NV hierarchy. We use the relaxation degree corresponding to monomials $\{\openone, \rho, M, M_\mathbf{povm}, \rho \rho\}$ and also all the monomials $\rho_xM_{(x,x')}^b$ appearing in the first summand of \eqref{witness}. In Table \ref{tablenoproj} we present the upper bounds $W_d^P$. We have also obtained lower bounds for $W_d$ under projective measurements by considering SDPs in alternate convex search, enforcing only $d$ non-zero elements of trace one. These lower bounds were verified to be achieved with projective measurements up to machine precision. The results show that the obtained upper bounds are either optimal or close to optimal, depending on $d$. In analogy with previous works \cite{APV16, GGG16, Armin, Piotr, Massi}, we find that the gap between optimal quantum correlations and those obtained under projective measurements is small.

Let us now consider the role of symmetrisation in obtaining the above results. In Table \ref{tablesizes} we present the number of samples needed to complete the basis in the NV hierarchy, the size of the final SDP matrix, and the time required to evaluate the SDPs. We compare these parameters for a standard implementation, a symmetrised implementation only reducing the number of samples, and a the full symmetrisation developed to also exploit block-diagonalisation of the SDP matrix. Without symmetries, we are unable to go beyond qubit systems ($d=2$), since already for $d=3$ we have over $12000$ samples. Interestingly, this rapid increase in complexity can be completely overcome via symmetrisation: the number of samples becomes constant when $d=4,5,6$. In addition, the size of the SDP matrix  is $1+d-2d^2+3d^4$ and thus increases polynomially in $d$. This causes a symmetrisation that only addresses the number of samples to still be too demanding already when $d>4$. However, using the block-diagonalisation methods detailed in SM, we can reduce the size of the SDP matrix to be constant for $d=4,5,6$. This allows us to straightforwardly solve the semidefinite relaxations in less than two seconds.

\begin{table}
	\begin{tabular}{l|r|ccccc}
		& d           & 2            & 3                   & 4                & 5                 & 6                \\ \hline
		& \#samples   & 221          & \textgreater{}12000 & -                & -                 & -                \\
		& bl. sizes   & 1{[}43{]}    & 1{[}229{]}          & 1{[}741{]}       & 1{[}1831{]}       & 1{[}3823{]}      \\
		\multirow{-3}{*}{\begin{tabular}[c]{@{}l@{}}Non-\\ sym\end{tabular}}  & SDP {[}s{]} & 2.0          & -                   & -                & -                 & -                \\ \hline
		& \#samples   & 65           & 134                 & \multicolumn{3}{c}{\cellcolor[HTML]{C0C0C0}137}         \\
		& bl. sizes   & 1{[}43{]}    & 1{[}229{]}          & 1{[}741{]}       & 1{[}1831{]}       & 1{[}3823{]}      \\
		\multirow{-3}{*}{\begin{tabular}[c]{@{}l@{}}Sym\\ no BD\end{tabular}} & SDP {[}s{]} & 0.5          & 19                  & 500               & -                 & -                \\ \hline
		& \#samples   & 65           & 134                 & \multicolumn{3}{c}{\cellcolor[HTML]{C0C0C0}137}         \\
		& bl. sizes   & 4{[}6,16{]} & 7{[}3,16{]}        & \multicolumn{3}{c}{\cellcolor[HTML]{C0C0C0}8{[}3,16{]}} \\
		\multirow{-3}{*}{\begin{tabular}[c]{@{}l@{}}Sym\\ +BD\end{tabular}}   & SDP {[}s{]} & 0.3          & 0.6                 & \multicolumn{3}{c}{\cellcolor[HTML]{C0C0C0}1.2}        
	\end{tabular}
	\caption{Comparison between computational parameters for the task of bounding $W_d$ under projective measurements using a standard implementation, symmetrisation to reduce the number of samples (using only Eq.~\eqref{Reynolds}), and symmetrisation to also perform block-diagonalisation (BD). The notation $D[a,b]$ means that there are $D$ blocks with the smallest being of size $a$ and the largest of size $b$.}\label{tablesizes}
\end{table}

	\textit{Further applications.---} The general symmetrisation technique can be used to a wide variety of problems in quantum information theory, among which certification of non-projective measurement constitutes one example. In SM, we consider in detail four families of other problems. For each, we  demonstrate the remarkable computational advantages of symmetrisation, both in terms of reducing the number of basis elements and in terms of block-diagonalisation. This enables us to obtain  improved bounds on previously studied physical quantities. The problems we consider are (high-dimensional and many-input) random access codes \cite{Ambainis, TavakoliRACs}, $I_{3322}$-like Bell inequalities \cite{Froissart, NV2}, a sequential communication in multipartite CCPs (in the spirit of \cite{Galvao, TS05}), and CCPs exhibiting dimensional discontinuities \cite{magic7, DM18}. In the latter, we also exemplify the advantages in automatising the search for the symmetries in problems in which these are not easily spotted by inspection. 
	
	Moreover, we previously observed that the complexity of the evaluation for bonuding $W_d^P$ can be reduced to be constant $d=4,5,6$ via symmetries. This suggests that similar reductions may occur for other CCPs as well.  In SM we have focused on the CCPs known as random access codes and proven that symmetries enable us to evaluate the NV hierarchy with constant complexity for any Hilbert space dimension. In this sense, the computational advantages over standard implementations, as well as over symmetrisation that does not utilise block-diagonalisation, increase with $d$.

\textit{Conclusions.---} We presented a technique for efficiently evaluating semidefinite relaxations of finite-dimensional quantum correlations using symmetries present in the problem. The technique provides remarkable computational advantages and applies to general dimension-bounded quantum correlation problems, which we demonstrated by explicit examples. In particular, we introduced CCPs that self-test $d$-dimensional SIC-POVMs and used them to certify the non-projectiveness of measurements close to SIC-POVMs. Due to the broad applications of SIC-POVMs in quantum information theory, such certificates are relevant to recent experimental advances in high-dimensional quantum systems. A relevant open problem is how to construct witnesses that allow for larger gaps between the projective measurement bound and the quantum bound.
	
We conclude with two open problems. Can the sampling approach be adapted to semidefinite relaxations in Bell inequalities without dimensional bounds? How does the symmetrisation technique adapt to physical problems that do not concern quantum resources; e.g., cardinality of hidden variables \cite{lhv} and the dimension of post-quantum resources?

\textit{Acknowledgements.---} During the completion of this work, we became aware of a work-in-preparation by E. Aguilar and P. Mironowicz to generalise the results of \cite{Sym8}. We are thankful for useful discussions with Jean-Daniel Bancal.
This work was supported by the Swiss National Science Foundation (Starting grant DIAQ, NCCR-QSIT). Research at Perimeter Institute is supported by the Government of Canada through Industry Canada and by the Province of Ontario through the Ministry of Research and Innovation. This publication was made possible through the support of a grant from the John Templeton Foundation.

\onecolumngrid
\appendix 

\section{Ambient groups and symmetry groups}\label{appAmbient}
Here, we describe the general construction of ambient groups for Bell scenarios and communication complexity problems (CCPs) computations. Building on these constructions, we present a simple manner of automatising a search for generators of the symmetry group of an optimisation problem. However, before those considerations, we give a short overview of the terminology and the problem.

The optimisation is conducted by evaluating a polynomial $p$ over a
set of states that depends on the problem. We express that state using its
Kraus decomposition $\rho = K^{\dag} K$:
\begin{equation}
\label{Eq:Problem} \max_{X, K} \tr [Kp (X) K^{\dag}]
\end{equation}
\[ \text{subject to } q_1 (X) = 0, \quad q_2 (X) = 0, \quad \ldots \]
such that $p (X)$ and $\{ q_j (X) \}$ are polynomials in the operator 
variables $(X_i)_{i \in \mathcal{I}}$, where $\mathcal{I}= \{ 1, \ldots, |
\mathcal{I} | \}$ \ and $p (X)$ is Hermitian.

We consider evaluating~(\ref{Eq:Problem}) with a specified finite dimensional
bound. A {\em feasible realisation} is given by a sequence of matrices
$\overline{X} = \left( \overline{X}_i \right)_{i \in \mathcal{I}}$ satisfying
the constraints $\left\{ q_1 \left( \overline{X} \right) = 0, \ldots \right\}$
and a finite dimensional $\overline{K}$ taken from a set $\mathcal{K}$
specified by the problem.
\begin{itemize}
	\item In all the CCPs considered, we use the
	tracial hierarchy of Burgdof and Klep~{\cite{BK}}, $\mathcal{K}= \{
	\mathbbm{1} / d \}$, in which the preparations are represented by density
	matrices that are absorbed into the variables $(X_i)$.
	
	\item In our $I_{3322} (c)$ example, we use the NPA hierarchy where
	$\mathcal{K}= \{ | \psi \rangle : | \psi \rangle \in \mathcal{H},
	\langle \psi | \psi \rangle = 1 \}$ and $\mathcal{H}$ is a finite
	dimensional Hilbert space.
	
	\item The hierarchy of Moroder et al.~{\cite{Moroder}} can be implemented by
	considering a set $\mathcal{K}= \left\{ \overline{K} : \rho \equiv \left(
	\overline{K} \cdot \overline{K}^{\dag} \right) \in \text{PPT} \right\}$
	where $\overline{K}$ is the Kraus decomposition of a positive partial-transpose state $\left(
	\rho^{\top_{\text{B}}} \geqslant 0 \right)$. We have not implemented this particular hierarchy. 
\end{itemize}
We can further restrict the feasible realisations $\overline{X}$ by requiring
the matrices $\left\{ \overline{X}_i \right\}$ to obey additional conditions, for example rank constraints. We write $\Xi = \left\{ \overline{X} \right\}$
the {\em feasible set} of realisations that obey the constraints $\{ q_1
(X) = 0, \ldots \}$ and rank-like constraints.
A symmetry of $\Xi$ is a permutation $\pi : \mathcal{I} \rightarrow
\mathcal{I}$ of the indices $\mathcal{I}$ that obeys
\begin{equation}
\label{Eq:Symmetry} \left( \overline{X}_i \right)_{i \in \mathcal{I}} \in
\Xi \qquad \Rightarrow \qquad \pi \left( \overline{X} \right) = \left(
\overline{X}_{\pi^{- 1} (i)} \right)_{i \in \mathcal{I}} \in \Xi,
\end{equation}
where the definition follows from the requirement $(\sigma \pi) \left(
\overline{X} \right) = \sigma \left( \pi \left( \overline{X} \right) \right)$.
We call the group of all permutations that preserve the structure of the problem~(\ref{Eq:Symmetry}) the {\em ambient group}, $\mathcal{A}= \{ \pi \}$. Similarly, we write $\mathcal{G} \subseteq \mathcal{A}$ the {\em symmetry group} of
the problem which additionally leaves the objective invariant:
\begin{equation}
\label{Eq:SymmetryGroup} \mathcal{G}= \left\{ \pi \in \mathcal{A} \quad :
\quad \forall \overline{K} \in \mathcal{K}, \overline{X} \in \Xi, \quad
\tr \left[ \overline{K}^{\dag} p \left( \overline{X} \right)
\overline{K} \right] = \tr \left[ \overline{K}^{\dag} p \left( \pi
\left( \overline{X} \right) \right) \overline{K} \right] \right\} .
\end{equation}

A final remark: we emphasise that $\mathcal{A}$ acts not on physical systems (or their labels), but rather on the abstract operator variables. This removes a source of confusion when constructing the symmetry group of the SDP relaxation. For example, in the RAC example from the main text, the re-labelling of the output $b$ cannot depend on $x$, as the operator $M^b_y$ does {\em not} have an $x$ index.

Next, we will consider the general construction of ambient groups for scenarios common in quantum information.
\subsection{Ambient groups in prepare-and-measure scenarios}

In the prepare-and-measure scenario outlined in the introduction, the set of operators has size $\mathcal{X} + \mathcal{B} \mathcal{Y}$. It is given by $\{\rho_x\} \cup \{M_y^b\}$ for the inputs $x=1,\ldots,\mathcal{X}$, the inputs $y=1,\ldots,\mathcal{Y}$ and the outputs $b=1,\ldots,\mathcal{B}$. We have the constraints
\begin{equation}
\rho_x \succeq 0, \quad M_y^b \succeq 0, \qquad \sum_b M_y^b = \mathbbm{1}\;,
\end{equation}
in addition to the generic constraints of the tracial moment hierarchy.

\begin{proposition}
	\label{prop:ambientprepmeas}
	In prepare-and-measure scenarios, elements of the ambient group $\mathcal{A}$ are uniquely enumerated by
	\[
	\pmb{a} = \pmb{\xi} ~ \pmb{\psi} ~ \pmb{\beta_1} \ldots \pmb{\beta_Y} \;,
	\]
	where $\pmb{\xi}$, $\pmb{\psi}$ and $\pmb{\beta_y}$ are permutation of the operators $\{ X_i \}$ defined as follows.
	\begin{itemize}
		\item The permutation $\pmb{\xi}$ corresponds to a re-labelling of the input $x$ and is parameterised by a permutation $\xi \in S_\mathcal{X}$.
		It acts as
		\[
		\pmb{\xi}(\rho_x) = \rho_{\xi(x)}, \qquad \pmb{\xi}(M_y^b) = M_y^b\;.
		\]
		\item The permutation $\pmb{\psi}$ corresponds to a re-labelling of the input $y$ and is parameterised by a permutation $\psi \in S_\mathcal{Y}$.
		It acts as
		\[
		\pmb{\psi}(\rho_x) = \rho_x, \qquad \pmb{\psi}(M_y^b) = M_{\psi(y)}^b\;.
		\]
		\item The permutation $\pmb{\beta_y}$ corresponds to a re-labelling of the output $b$ conditioned on the input $y$ and is parameterised by a permutation $\beta_y \in S_\mathcal{B}$.
		It acts as
		\[
		\pmb{\beta_y}(\rho_x) = \rho_x, \qquad \pmb{\beta_y}(M_y^b) = M_y^{\beta_y(b)}, \qquad \pmb{\beta_y}(M_{y'}^b) = M_{y'}^b \text{ if } y \ne y'\;.
		\]
	\end{itemize}
	The ambient group $\mathcal{A}$ has order $\mathcal{X}! ~ \mathcal{Y}!(\mathcal{B}!)^\mathcal{Y}$
\end{proposition}
\begin{proof}\textit{(Sketch)}
	Due to the normalisation constraint, a valid permutation $\pmb{a}\in S_{\mathcal{X}+ \mathcal{B} \mathcal{Y}}$ cannot permute a state $\rho_x$ into a measurement $M_y^b$. Moreover, permutations of measurements have to preserve the block structure given by $\{M_1^b\},\ldots,\{M_\mathcal{Y}^b\}$.
	Thus, the ambient group is given by $\mathcal{A} = \mathcal{S} \times \mathcal{M}$, where $\mathcal{S}$ represents arbitrary permutations of states $\{\rho_x\}$ and $\mathcal{M}$ represents permutations of measurements.
	The group $\mathcal{S}$ is isomorphic to $S_\mathcal{X}$, the symmetric group of degree $\mathcal{X}$, which has order $\mathcal{X}!$.
	Elements of the group $\mathcal{M} = \{m\}$ can uniquely be written as the product of a permutation of inputs $\pmb{\psi}$, parameterised by $\psi \in S_\mathcal{Y}$, and permutations of outputs $\pmb{\beta_1}$, \ldots, $\pmb{\pi_\mathcal{Y}}$, parameterised by $\beta_y \in S_\mathcal{B}$. $\mathcal{M}$ has order $\mathcal{Y}! (\mathcal{B}!)^\mathcal{Y}$.
\end{proof}

Formally, the group $\mathcal{M}$, which preserves the block structure, is a {\em wreath product} of $S_\mathcal{B}$ by $S_\mathcal{Y}$~\cite{Renou2016}.

\subsection{Ambient groups in Bell scenarios}

For simplicity, we consider two-party Bell scenarios, which are written using the operators $\{A_{a|x}\}$ and $\{B_{b|y}\}$, for inputs $x,y=1,\ldots,m$ and outputs $a,b=1,\ldots,d$. This can easily be generalised to more parties. The constraints are:
\begin{equation}
A_{a|x} \succeq 0, \quad B_{b|y} \succeq 0, \qquad \sum_a A_{a|x} = \mathbbm{1}, \quad \sum_b B_{b|y} = \mathbbm{1}\;.
\end{equation}

\begin{proposition}
	In Bell scenarios, any valid permutation of operators is uniquely written
	\[
	\pmb{a} = \pmb{\xi} ~ \pmb{\psi} ~ \pmb{\alpha_1} \ldots \pmb{\alpha_m} ~ \pmb{\beta_1} \ldots \pmb{\beta_m}
	\quad
	\text{ or }
	\quad
	\pmb{a} = \pmb{\pi} ~ \pmb{\xi} ~ \pmb{\psi} ~ \pmb{\alpha_1} \ldots \pmb{\alpha_m} ~ \pmb{\beta_1} \ldots \pmb{\beta_m}
	\]
	where $\pmb{\pi}$ represents the swap of parties, $\pmb{\xi}$, $\pmb{\psi}$ are permutations of inputs and $\pmb{\alpha_x}$, $\pmb{\beta_y}$ are permutations of outputs with the following definitions.
	\begin{itemize}
		\item The permutation $\pmb{\pi}$ acts as:
		\[
		\pmb{\pi}(A_{a|x}) = B_{a|x}, \qquad \pmb{\pi}(B_{b|y}) = A_{b|y}\;.
		\]
		\item The permutation $\pmb{\xi}$ corresponds to a re-labelling of the input $x$ and is parameterised by a permutation $\xi \in S_m$.
		It acts as
		\[
		\pmb{\xi}(A_{a|x}) = A_{a|\xi(x)}, \qquad \pmb{\xi}(B_{b|y}) = B_{b|y}\;.
		\]
		\item The permutation $\pmb{\psi}$ corresponds to a re-labelling of the input $y$ and is parameterised by a permutation $\psi \in S_m$.
		It acts as
		\[
		\pmb{\psi}(A_{a|x}) = A_{a|x}, \qquad \pmb{\psi}(B_{b|y}) = B_{b|\psi(y)}\;.
		\]
		\item The permutation $\pmb{\alpha_x}$ corresponds to a re-labelling of the output $a$ conditioned on the input $x$ and is parameterised by a permutation $\alpha_x \in S_d$.
		It acts as
		\[
		\pmb{\alpha_x}(A_{a|x}) = A_{\alpha_x(a)|x}\qquad \pmb{\alpha_x}(A_{a|x'}) = A_{a|x'} \text{ if } x \ne x', \qquad \pmb{\alpha_x}(B_{b|y}) = B_{b|y}\;,
		\]
		\item The permutation $\pmb{\beta_y}$ corresponds to a re-labelling of the output $b$ conditioned on the input $y$ and is parameterised by a permutation $\beta_y \in S_d$.
		It acts as
		\[
		\pmb{\beta_y}(A_{a|x}) = A_{a|x}, \qquad \pmb{\beta_y}(B_{b|y}) = B_{\beta_y(b)|y}, \qquad \pmb{\beta_y}(B_{b|y'}) = B_{b|y'} \text{ if } y \ne y'\;.
		\]
	\end{itemize}
	The ambient group $\mathcal{A}$ has order $2 (m!)^2 ~ (d!)^{2m}$.
\end{proposition}
\begin{proof} \textit{(Sketch)}
	Due to the normalisation, we need to preserve a two-level block structure. First, we can permute measurements of Alice and Bob provide we permute {\em all} of them.
	This corresponds to permutation of parties, a symmetry that has already been used in the literature~\cite{Moroder}.
	Then, we have two groups: the first one acts on the measurements of Alice only; the second one on the measurements of Bob only. The action of those groups on the set of concerned operators is exactly the same as in the prepare-and-measure case.
	For additional details about the symmetry groups of Bell scenarios, see \cite{50years, Renou2016}.
\end{proof}

Remark that to construct the ambient group for $n>2$ parties, we simply parameterise $\pmb{\pi}$ by an arbitrary permutation of parties taken from $S_n$ and add additional elements in the decomposition of $\pmb{a}$ corresponding to permutations of inputs/outputs of the additional parties.
The resulting group is then a {\em double} wreath product, of $S_d$ by $S_m$ by $S_n$ (see again~\cite{Renou2016}).

\section{Software package for symmetrisation: theory and practice}\label{appPackage}
We make our symmetrisation tools publicly available in a user-friendly manner by providing a MATLAB package. The package applies to all problems of the form \eqref{Eq:Problem}, in particular general Bell scenarios and distributed computations (not necessarily limited to two parties).   Relying on randomised sampling, it requires the following information from the user.

\begin{itemize}
	\item A random oracle that returns a generic sample of the operator products $\overline{X} \in \Xi$.
	
	\item A random oracle that returns a generic sample of the Kraus operator $\overline{K} \in
	\mathcal{K}$.
	
	\item A black box function $f \left( \overline{X}, \overline{K} \right)$
	that computes the objective $\tr \left[ \overline{K}^{\dag} p \left( \overline{X}
	\right) \overline{K} \right]$, as given in \eqref{Eq:Problem}.
	
	\item A bound $L$ on the degree of products of operators in the hierarchy,
	with the constraint that $p (X)$ has monomials of degree at most $2 L$.
	
	\item The generators of the symmetry group $\mathcal{G}$.
\end{itemize}
The user does not need to specify the constraints $\{ q_1 (X) = 0, \ldots \}$, but rather implement an oracle that samples realisations generically from the feasible set. If these constraints are provided, the package will use them to validate the symmetry group.

Our algorithm outputs a basis $(E_0, \{ E_1, \ldots, E_m \})$ of moment matrices in a block-diagonal basis, along with a real vector $\vec{b}$ such that the canonical semidefinite program
\begin{equation}
\label{Eq:SDPFormulation} \max_{\vec{y} \in \mathbbm{R}^m} \vec{b}^{\top} \cdot \vec{y} + b_0
\end{equation}
\[ \text{subject to} \qquad E_0 + \sum_{\ell = 1}^m y_{\ell} E_{\ell} \geq 0.
\]
provides an upper bound on the objective of the problem~(\ref{Eq:Problem}) under dimension (and possibly rank) constraints.

In the above, we assumed that the generators of the symmetry group $\mathcal{G}$ are known. The algorithm also works when a subset of those generators are provided, with a loss of efficiency --- when no generators are provided, our algorithm reduces to the standard NV hierarchy. However, if the ambient group~$\mathcal{A}$ is known instead, we provide a function that recovers the symmetry group from it, provided the size of $\mathcal{A}$ is small (say a few millions), as we simply filter the elements one by one. However, computing the symmetry group on a small representative of a problem can help the user to guess the form of the symmetry group for the general problem. This was exemplified in Example 2 of the main text, where only the cyclic symmetry can be immediately guessed.

Even when no symmetrisation is performed, our implementation improves on the original proposal of the NV hierarchy: we remove redundant monomials from the generating set, compute the samples in batches and pre-compute the contractions of monomials/Kraus operators.

\subsection{Four methods of symmetrisation}\label{secsec}

As seen in the main text, symmetrisation reduces the size of the basis. Afterwards, one can also block-diagonalise the moment matrix by a variety of techniques. In view of this, the MATLAB package is made available with four different symmetrisation methods (and one non symmetrised variant) :
\begin{itemize}
	\item {\ttfamily none:} Does not apply symmetrisation.
	\item {\ttfamily reynolds:} Averages the samples over the symmetry group by computing the Reynolds operator. This reduces the number of scalar variables in the SDP. It performs no block-diagonalisation.
	
	\item {\ttfamily isotypic:} In addition to reducing the number of scalar variables in the SDP via the Reynolds operator, it  identifies a partial block
	structure in $\{ E_0, E_1, \ldots \}$ (without multiplicities) after sampling and uses this to reduce the size of the positivity constraints.
	
	\item {\ttfamily irreps:} In addition to reducing the number of scalar variables in the SDP via the Reynolds operator, it decomposes the
	column space of $\{ E_0, E_1, \ldots \}$ into irreducible representations and performs a full block-diagonalisation after the samples are collected.
	
	\item {\ttfamily blocks:} Computes the irreducible representations of the symmetry group and uses these to sample directly in the block-diagonal basis, using an optimised version of the Reynolds operator.
\end{itemize}
Among these four methods, {\ttfamily reynolds} is the most elementary form of symmetrisation whereas {\ttfamily blocks} exploits the full potential of the symmetrisation technique.

\begin{table}[]
	\begin{tabular}{|l|l|l|l|l|l|l|l|l|l|l|}
		\hline
		& \multicolumn{5}{l|}{Dimension d = 3}                                                                        & \multicolumn{5}{l|}{Dimension d = 7}            \\ \hline
		Method   & Max. block size & Time: dec.. & ..basis & ..solver & Precision                                        & Max. block size  & Time: dec.. & ..basis & ..solver & Precision                           \\ \hline
		
		{\ttfamily none} & 70 &  & 0.1 & 1.0 & $3\cdot 10^{-13}$ & 750 & \multicolumn{4}{c|}{\text{\em Out of reach}}  \\ \hline
		{\ttfamily reynolds} & 70 & $0.04$ & $0.001$ & $0.1$ & $3\cdot 10^{-13}$
		& 750 & $2.8$ & $0.1$ & $61$ & $8\cdot 10^{-11}$ \\ \hline
		{\ttfamily isotypic} & 28 & $0.07$ & $0.0007$ & $0.04$ & $2\cdot 10^{-11}$
		& 180 & $4.8 ~ (3.9)$ & $0.01$ & $1.6$ & $1\cdot 10^{-9}$ \\ \hline
		{\ttfamily irreps} (N) & 7 & $0.09$ & $0.0004$ & $0.01$ & $5\cdot 10^{-13}$
		& 7 & $7.0 ~ (4.9))$ & $0.0004$ & $0.008$ & $1\cdot 10^{-10}$ \\ \hline
		{\ttfamily blocks} (N) & 7 & $0.08$ & $0.0004$ & $0.01$ & $2\cdot 10^{-13}$
		& 7 & $6.9 ~ (4.8)$ & $0.0004$ & $0.009$ & $7\cdot 10^{-12}$ \\ \hline
		{\ttfamily irreps} (A) & 7 & $0.05$ & $0.0004$ & $0.009$ & $1\cdot 10^{-12}$
		& 7 & $4.4 ~ (2.7)$ & $0.0006$ & $0.008$ & $1\cdot 10^{-10}$ \\ \hline
		{\ttfamily blocks} (A) & 7 & $0.05$ & $0.0003$ & $0.009$ & $7\cdot 10^{-14}$
		& 7 & $4.4 ~ (0.004)$ & $0.0004$ & $0.009$ & $2\cdot 10^{-12}$ \\ \hline
	\end{tabular}
	\caption{\label{tabImplementations} 
		Comparison of five implementations for the random access code (see section~\ref{RACs}), for $d=3,7$ and $n=2$, where each problem was solved 20 times. For the irreducible decomposition, we either used our numerical algorithm (N), or the analytical decomposition (A) provided in section~\ref{appIrreps}. All times are given in seconds: ``dec.'' corresponds to the construction and decomposition of the symmetry group and the numerical block-diagonalisation (in parentheses, consistency checks disabled), ``basis'' both to the computation of symmetrised moment matrices and the rank verification, ``solver'' to the time spent in the semidefinite programming solver, while times spent in the toolbox YALMIP are not presented. The precision is the average absolute deviation with respect to the correct objective. For this problem, we used MOSEK with a tolerance $\varepsilon = 0$, forcing the solver to iterate until no further progress is made.}
\end{table}

In Table~\ref{tabImplementations}, we compare the  five methods on the random access code (RAC) example of section~\ref{RACs}; this problem was already considered for the case of $n=2$ in~\cite{Sym8}, where the method they present corresponds to {\ttfamily reynolds}. Moreover, the cited work provides the analytical maximal value of $\mathcal{A}_{2,d}^{\text{RAC}}$ which we use the evaluate the numerical precision of our bounds.

Let us comment the impact of the successive refinements of our technique. First of all, symmetrisation of the moment matrix ({\ttfamily reynolds}) provides a large gain: it allows us to compute bounds for problems that were out of reach previously (such as our RAC example for $n=2$, $d=7$) are now within reach. 
Note that doing so only involves standard arithmetic (addition and multiplication), so no precision loss is observed on average. This step reduces the number of basis elements, but does not reduce the size of the blocks of the moment matrix. 
The next step is to block-diagonalise partially ({\ttfamily isotypic}) the moment matrix using the simple heuristic described in the main text. 
Doing so improves the computation time by an order of magnitude, at the price of a decrease in precision: both the basis construction and the solver efficiency is increased. We understand the loss of precision as coming from the computation of matrix eigenspaces.  We now move to the finest decompositions available ({\ttfamily irreps}, {\ttfamily blocks}). There, we compare the numerical basis obtained using our numerical algorithm and the analytical decomposition presented in section~\ref{appIrreps}.
In the $d=7$ example, we gain several orders of magnitudes in efficiency: this is not surprising as the final block sizes become independent of the dimension (see section~\ref{appIrreps} for a discussion of these block sizes).
We also regain some precision, to the point that the fully block-diagonalised problem provides increased precision compared to the less symmetrised variants: this can be due to a special refinement step that we incorporated in the decomposition algorithm explained in~\cite{Denis}.
Out of the two variants presented, {\ttfamily blocks} performs less arithmetic operations and provides a precision advantage as a result.
Note the existing literature~\cite{NV, NV2} did not address numerical precision, a problem we will consider in future work~\cite{Denis}.

We stress that we did not optimise the MATLAB implementation of our algorithms for group/representation decomposition and that by default the code performs safety checks at every step.
This explains why, for example, no gains in overall processing time are obtained going from {\ttfamily isotypic} to {\ttfamily irreps (N)} with checks enabled, or why we spend time performing a group decomposition when an explicit basis is provided ({\ttfamily irreps (A)}).
We present the timings with safety checks removed in parentheses, although we do not recommend the use of our software in that manner.

We now turn to the $I_{3322}(c=1)$ example of section \ref{Bell}, where the bound for qubits is known~\cite{NV2} to be $5$ up to machine precision and perform the same tests on that new problem. The results are presented in Table~\ref{tabImplementationsI3322}. Compared to the RAC example, where the symmetry group was big, the $I_{3322}$ inequality only has a symmetry group of order $8$; this translates as smaller decreases in block sizes. Here, using either {\ttfamily irreps} or {\ttfamily blocks} is always worthwhile in terms of precision and total computation time.


Finally, we remark that our block-diagonalisation method decomposes representations over the reals. Three types of irreducible representations appear, either real, complex or quaternionic. We present below in full detail the case of real representations, which is sufficient to handle all examples presented in this manuscript.
Our code also implements the decomposition of complex representations, and all symmetrization methods are supported for real and complex-type representations.
Adding support for quaternionic representations (which occur very infrequently) is left open; in that case, coarser methods such as "isotypic" should be used.

\begin{table}[]
	\begin{tabular}{|l|l|l|l|l|l|l|l|l|l|l|}
		\hline
		& \multicolumn{5}{l|}{NPA level 2 + AAA + BBB}                                                                        & \multicolumn{5}{l|}{NPA level 4}            \\ \hline
		Method   & Max. block size & Time: dec.. & ..basis & ..solver & Precision                                        & Max. block size  & Time: dec.. & ..basis & ..solver & Precision                           \\ \hline
		{\ttfamily none} & 52 & 0.04 & 0.04 & 0.5 & $1\cdot 10^{-10}$ & \multicolumn{5}{c|}{\text{\em Too slow}}  \\ \hline
		{\ttfamily reynolds} & 52 & $0.02$ & $0.003$ & $0.2$ & $8\cdot 10^{-11}$
		& 244 & $0.06$ & $1.14$ & $44$ & $5\cdot 10^{-10}$ \\ \hline
		{\ttfamily isotypic} & 26 & $0.04$ & $0.002$ & $0.08$ & $1\cdot 10^{-11}$
		& 122 & $0.2$ & $0.4$ & $9.3$ & $5\cdot 10^{-10}$ \\ \hline
		{\ttfamily irreps} (N)& 13 & $0.05$ & $0.002$ & $0.04$ & $8\cdot 10^{-12}$
		& 61 & $0.4$ & $0.2$ & $2.8$ & $5\cdot 10^{-10}$ \\ \hline
		{\ttfamily blocks} (N)& 13 & $0.05$ & $0.002$ & $0.04$ & $1\cdot 10^{-11}$
		& 61 & $0.4$ & $0.2$ & $2.9$ & $1\cdot 10^{-9}$ \\ \hline
	\end{tabular}
	\caption{\label{tabImplementationsI3322}
		Comparison of our five implementations for the $I_{3322}(c = 1)$ inequality (see also section~\ref{Bell}), for qubits and rank-1 projectors. To apply our method on small and medium-size relaxations, we used two different hierarchy levels. Column legends are the same as in Figure~\ref{tabImplementations} and apart from the problem, the computation settings are the same. Note that we did not compute an analytical decomposition of the group representation.}
\end{table}

\subsection{Improvements not related to symmetries}

We first discuss the non symmetrised variant {\ttfamily none}, as the other methods are based on it.
We pay special attention to the places where our implementation differs from the one presented in~\cite{NV,NV2}.

\subsubsection{Monomial generating set}
To construct the moment matrix $\Gamma$ from a sample $\left( \overline{X}, \overline{K} \right)$, we need to determine a list $(s_1 (X), \ldots, s_n (X))$ of products of operators $s_j = X_{i_1} X_{i_2} \ldots$ such that
\[
\Gamma_{j, k} = \tr \left[ \overline{K}^{\dag} ~ s_j (\overline{X} )^{\dag} ~ s_k( \overline{X} ) ~ \overline{K} \right],
\]
where all products of at most $L$ operators appear.
For numerical stability and group action identification purposes, we require $\{ s_j (X) \}$ to be duplicate-free.
For that purpose, we generate all possible products of a most $L$ operators and evaluate $s_j \left( \overline{X} \right)$ using a generic sample $\overline{X} \in \Xi$, keeping a single representative for each set of indices $\{ j_1, j_2, \ldots \}$ for which $s_{j_1} ( \overline{X} ) = s_{j_2} ( \overline{X} ) = \ldots$.
A small optimisation is to remove the duplicates at each step, generating sets of products of degree $2$, $3$, ..., until $L$ iteratively by adding a single element in the products.
From now on, we call $\{ s_1, \ldots, s_n \}$ the {\em monomial generating set} with each $s_j$ a monomial of degree at most $L$ and denote the indices of the $\{ s_j \}$ by $j \in \mathcal{J}= \{ 1, \ldots, n \}$.

\subsubsection{Sampling algorithm and consistency check}

We are now ready to describe the naive implementation of our symmetrisation algorithm.
As a parameter, it requires a block size $B$.
\begin{algorithm}[H]
	\caption{Computing a basis of the moment matrix subspace numerically}
	\label{algSampling}
	\begin{algorithmic}
		\State $\ell \leftarrow 0$
		\Repeat
		\For{$i=1,...,B$}
		\Comment{Compute a batch of samples, can be parallelised.}
		\State $\ell \leftarrow \ell + 1$.
		\State Sample $\overline{X}$ and $\overline{K}$ using the oracle.
		\For{$j \in \mathcal{J}$}
		\Comment{Precompute monomial-Kraus operator products.}
		\State $\hat{s}_j \leftarrow s_j \left( \overline{X} \right) \overline{K}$.
		\EndFor
		\For{$j,k \in \mathcal{J}$}
		\Comment{Compute the moment matrix elements.}
		\State $\Gamma^{(\ell)}_{j, k} \leftarrow \tr [\hat{s}_j^{\dag}  \hat{s}_k]$.
		\EndFor
		\State $p^{(\ell)} \leftarrow f \left( \overline{X}, \overline{K} \right)$.
		\Comment{Compute the objective value.}
		\EndFor
		\State $r \leftarrow \text{rank} \{ \Gamma^{(1)}, \ldots, \Gamma^{(\ell)} \}$
		\Comment{Rank test.}
		\Until{$\ell > r$}
	\end{algorithmic}
\end{algorithm}

At the end of the algorithm, the set $\{ \Gamma^{(1)}, \ldots, \Gamma^{(r)} \}$ provides a basis for the feasible affine space $\mathcal{F}$ of
moment matrices.
We then set $E_0 = \Gamma^{(1)}$, $E_{\ell} = \Gamma^{(\ell + 1)} - \Gamma^{(1)}$, $b_0 = p^{(1)}$ and $b_{\ell} = p^{(\ell + 1)} - p^{(1)}$ in the SDP formulation~(\ref{Eq:SDPFormulation}).
By construction, we have an extra sample $\Gamma^{(r + 1)}$ which we use for a consistency check.
As the space $\mathcal{F}$ is of rank $r$, there is a set of coefficients $\vec{c} \in \mathbbm{R}^r$ such that
\[
\Gamma^{(r + 1)} = \sum_{\ell = 1}^r c_{\ell} \Gamma^{(\ell)}   .
\]
By construction, the objective function depends linearly on the moment matrix. Thus we verify that
\[
p^{(r + 1)} = \sum_{\ell = 1}^r c_{\ell} p^{(r)}
\]
up to a tolerance $\varepsilon$.
If the test fails, it either means that the numerical precision is insufficient for the problem size, or that the upper bound $L$ on the degree is insufficient for the given objective.

\subsubsection{Efficiency improvements}

In all the cases considered in this manuscript (and most applications), every feasible moment matrix $\Gamma \in \mathcal{F}$ has its complex conjugate feasible as well, $\Gamma^{\ast} \in \mathcal{F}$.
In that case, we can replace any solution $\Gamma$ by the real part $\Re [\Gamma] = (\Gamma + \Gamma^{\ast}) / 2$, which we can do directly during sampling.
We also pre-compute the products $s_j \left( \overline{X} \right) \overline{K}$, which leads to a small gain of efficiency, in particular for problems involving pure states $\overline{K} = | \psi \rangle$.
For problems involving medium-sized sets of samples, we found the Gram-Schmidt orthonormalisation slower than rank computations.
Thus, we iteratively compute sets of additional samples of fixed size and add them to the basis in batches.
After each addition, we compute the rank of the new sample space until the basis is saturated, at which point we truncate it to the correct number of samples.
The optimal value of the number of samples $B$ per batch depends on the problem (in our examples, we used $B = 100$ as a starting point).
In any case, we want to use as little arithmetic as possible on the samples to minimise the loss of precision.

For the computation of the Bell inequality bounds, we considered separately different combinations of ranks for the projective measurements (remark that now the rank corresponds to the operator variables and not to the rank of the moment matrix as above).
To optimise the process, we can quickly rule out deterministic measurements (corresponding to degenerate projectors) by doing the following.
We fix, in turn, a single projector to be deterministic by direct modification of the objective polynomial and then compute the quantum bound of the inequality without dimension constraints.
When the resulting bound is lower than the best known quantum model, those deterministic projectors can safely be omitted in the search.
For some variants of $I_{3322}(c)$ (see section~\ref{Bell}) in dimension 4, this reduces the number of cases from $5^6 = 15625$ to $3^6 = 729$.

\subsection{Symmetrisation via {\ttfamily reynolds}}

The simplest form of symmetrisation amounts to identifying a number of symmetries and reducing the number of linearly independent sampled matrices in the NV hierarchy, without considering the possibility of block-diagonalisation.
This type of symmetrisation corresponds to the method {\ttfamily reynolds} in the presented MATLAB package.

\subsubsection{Permutations of monomials and symmetrisation}

Let $\pi \in \mathcal{G}$ be a symmetry of the problem, which acts on the
index set $\mathcal{I}$ of the operators $\{ X_i \}$. For a monomial $s =
X_{i_1, i_2, \ldots}$, we defined the action of $\mathcal{G}$ on $s$ as $\pi
(s) = X_{\pi (i_1), \pi (i_2), \ldots}$. As the degree of $s$ does not
increase under symmetry, for each monomial $s_j$ in the monomial set, there is
another monomial $s_{j'} = \pi (s_j)$ in that set. Thus, $\pi : \mathcal{I}
\rightarrow \mathcal{I}$ corresponds to a permutation $\varphi (\pi) :
\mathcal{J} \rightarrow \mathcal{J}$ of the monomial indices $\mathcal{J}$.
Before running our sampling, we pre-compute all images $\varphi (\mathcal{G}) =
\{ \varphi (\pi) : \pi \in \mathcal{G} \}$, so that the action of
$\mathcal{G}$ on $\Gamma$, with image $\pi (\Gamma)$, is written
\begin{equation}
\label{Eq:PermutationMatrix}
\pi (\Gamma) = M_{\pi} \Gamma M_{\pi}^{\dag}, \qquad (M_{\pi})_{j, k} =
\left\{ \begin{array}{ll}
1 & \text{if } j = [\varphi (\pi)] (k)\\
0 & \text{otherwise} .
\end{array} \right.
\end{equation}
where $M_{\pi}$ is a permutation matrix.  Now, given a moment matrix $\Gamma$, we compute its symmetrisation
$\Gamma' =\mathcal{R}_{\mathcal{G}} (\Gamma)$ as
\[
\Gamma' = \frac{1}{| \mathcal{G} |} \sum_{\pi \in \mathcal{G}} \pi (\Gamma) \;,
\]
and store $\Gamma'$ instead of $\Gamma$ in the sequence of samples.

\subsubsection{Identifying the symmetry group}

In case little, or nothing, is known about the group $\mathcal{G}$, one may resort to searching for symmetries using only the group $\mathcal{A}$ and randomised sampling,
replacing the definition~(\ref{Eq:SymmetryGroup})
\[ \mathcal{G}= \left\{ \pi \in \mathcal{A} \quad : \quad \tr \left[
\overline{K}^{\dag} p \left( \overline{X} \right) \overline{K} \right] =
\tr \left[ \overline{K}^{\dag} p \left( \pi \left( \overline{X}
\right) \right) \overline{K} \right] \right\} \]
for a single generic sample $\overline{X} \in \Xi$ and $\overline{K} \in
\mathcal{K}$. If necessary, the resulting group elements of $\mathcal{G}$ can
be checked for consistency by checking that they leave the objective invariant
for a second generic sample. This brute force approach is feasible for groups
$\mathcal{A}$ of size up to a few millions.

For bigger problems, an approach based on the permutation group algorithms
described in~{\cite{Holt2005}} can be used, but is not currently implemented.
We take the set of monomials present in $p (X)$ and complement it with their orbits under $\mathcal{A}$, removing duplicates from the result. Then we take a generic sample and evaluate those monomials in a vector $\vec{v}$, and
compute $\mathcal{G}$ as the subgroup of $\mathcal{A}$ that leaves $\vec{v}$
invariant up to some tolerance; this corresponds to the computation of a
partition stabiliser which can be performed efficiently for very large groups.

\subsubsection{Speeding up the computation of the Reynolds operator}
\label{Sec:SpeedingUp}

When $\mathcal{G}$ is large, a lot of time will be spent in the computation of the
sum~$\frac{1}{| \mathcal{G} |} \sum_{\pi \in \mathcal{G}} \pi (\Gamma)$. We
describe now a first way to speed it up. We call a {\em product decomposition} of
the group $\mathcal{G}$ a sequence of subsets $U_1, U_2, \ldots U_C$, so that
every element $\pi \in \mathcal{G}$ is uniquely written
\[ \pi = u_1 u_2 \ldots u_C, \qquad u_1 \in U_1, u_2 \in U_2, \ldots, u_C \in
U_C . \]
Following~\cite[Alg. 3.1.1]{Derksen2002}, the computation of the Reynolds
operator then reduces to
\begin{equation}
\label{ReynoldsImproved}
\mathcal{R}_G (\Gamma) = \frac{1}{| \mathcal{G} |} \sum_{u_1 \in U_1} M_{u_1} \left[ \sum_{u_2 \in U_2} M_{u_2} \left[ \ldots \left[ \sum_{u_C \in U_C} M_{u_C} \Gamma M_{u_C}^{\dag} \right] \ldots \right] M_{u_2}^{\dag} \right] M_{u_1}^{\dag},
\end{equation}
by linearity as $(u_1 u_2 \ldots u_n) (\Gamma) = u_1 (u_2 (\ldots (u_n(\Gamma))))$.
As $\mathcal{G}$ is a permutation group, a good decomposition is obtained by
computing a chain of stabilisers
\[ \mathcal{G} \supseteq \mathcal{G}_{(1)} \supseteq \mathcal{G}_{(1, 2)}
\supseteq \ldots \supseteq \mathcal{G}_{(1, 2 \ldots, | \mathcal{I} |)} \]
where $\mathcal{G}_S = \{ \pi \in \mathcal{G}: \forall i \in S, g (i) = i \}$
is the subgroup that fixes every index in $S$. We then take sets $\{ U_c \}$
from the coset transversals (see~{\cite{Holt2005}}). This computation can be
done efficiently from the generators of $\mathcal{G}$ using the randomised
Schreier-Sims algorithm~{\cite{Holt2005, Leon}}. These matters will be discussed in a future work~\cite{Denis}.

\subsubsection{Improvements to rank-constrained problems}
As a prerequisite, our symmetrisation method requires that if $\overline{X}$ is a sample, then $\pi(\overline{X})$ is a sample as well.
Thus, when considering rank constraints, we sample not only from a particular rank sequence, but also from all its permutations under the symmetry group.
For the $I_{3322}(c)$ example (see section~\ref{Bell}), our operators are $(X_1,X_2,X_3,X_4,X_5,X_6)= (A_1,A_2,A_3,B_1,B_2,B_3)$ and the rank sequence $\overline{r} = (r_1,r_2,r_3,r_4,r_5,r_6)$ corresponds to the number of eigenvalues equal to $+1$ for each of the measurements.
We remark that sampling from operators with rank sequence $\overline{r} = (2,2,2,1,1,1)$ is equivalent to sampling from operators with rank sequence $\overline{r} = (1,1,1,2,2,2)$ (for example) due to symmetries in the objective polynomial (here invariance under party permutation).
Thus, we only consider a single representative from the orbits of rank sequences under the symmetry group of the problem.

\subsection{Block-diagonalisation: elements of theory}
We start by reviewing the relevant mathematical notions: for a short introduction to the linear representation theory of finite groups, the reader can follow~{\cite{Serre1977}, see also \cite[Sec. 4]{Parrilo} for a summary of the notion applied to semidefinite programming.
	To match the formulation handled by most semidefinite programming solvers~\cite{complexsolvers}, including MOSEK~\cite{mosek}, we assume that the moment matrix $\Gamma$ is real and symmetric.
	Fortunately, this corresponds to most applications of moment relaxations in quantum information and to all examples presented in this manuscript. 
	In the rare case where a complex Hermitian $\Gamma$ is required, we assume that its reformulation as a real symmetric matrix~\cite[Ex. 4.42]{Boyd2004} has been done beforehand; the material below can then easily be adapted.
	
	We recall that the column space of the moment matrix $\Gamma$ is indexed by the monomials of the generating set $\mathcal{J}$.
	We write $V = \mathbb{R}^{|\mathcal{J}|}$ the column space of the moment matrix.
	Given a permutation $\pi\in\mathcal{G}$ of the operator variables, we defined in~\eqref{Eq:PermutationMatrix} the action of $\pi$ on $V$, which we wrote as a permutation matrix $M_\pi$.
	\subsubsection{Isotypic decomposition}
	From group representation theory, we know that there exists a change of basis matrix $U_\text{iso}$, so that the permutation matrix $M_\pi$ has the block diagonal form
	\[
	\tilde{M}_{\pi, \text{iso} } =  U^\top_\text{iso} M_\pi U_\text{iso} =
	\tilde{M}^1_{\pi, \text{iso}} \boxplus \ldots \tilde{M}^R_{\pi \text{iso}}
	\qquad \text{defining} \qquad X \boxplus Y =
	\begin{pmatrix} X & \\ & Y \end{pmatrix}\;,
	\]
	for arbitrary $\pi \in \mathcal{G}$, where the blocks $\tilde{M}^r_{\pi, \text{iso}}$ correspond to a decomposition of the vector space $V$:
	\begin{equation}
	\label{blkdiag_space}
	V = W^1 \oplus W^2 \oplus \ldots \oplus W^R\;,
	\end{equation}
	with the restriction that each {\em isotypic component} $W^r$ contains copies of a unique irreducible representation, for $R$ inequivalent irreducible representations ({\em irreps}).
	The block-diagonal form $\tilde{M}_{\pi, \text{iso}}$ highlights invariant subspaces of $V$.
	The basis vectors of these components form the columns of $U_\text{iso}$:
	\[
	U_\text{iso} = \left( \vec{w}^1_1, \ldots, \vec{w}^1_{\dim W^1}, \ldots, \vec{w}^R_1, \ldots \vec{w}^R_{\dim W^R}  \right) \;.
	\]
	so that $\big \{ \vec{w}^{r}_i \big \}$ are orthonormal basis vectors such that $W^r = \Span \big \{ \vec{w}^{r}_1, \ldots \vec{w}^r_{\dim W^r} \}$.
	The decomposition of $V$ into $\{ W^r \}$ is called the {\em isotypic decomposition} (see~{\cite[Sec. 2.6]{Serre1977}}) and is unique; the basis given by $U_\text{iso}$ is called the {\em isotypic basis}.
	It is a coarse-graining of the irreducible decomposition presented in the next section.
	The basis vectors are defined up to a unitary change of basis inside each component $W^r$.
	\subsubsection{Isotypic decomposition: impact on invariant symmetric matrices}
	We consider a real matrix $\Lambda \in \mathbb{R}^{|\mathcal{J}| \times |\mathcal{J}|}$ which satisfies:
	\begin{equation*}
	\label{Eq:CondLambda}
	\Lambda^\top = \Lambda, \qquad M_\pi^\top \Lambda M_\pi = \Lambda, \quad \forall \pi \in \mathcal{G}\;,
	\end{equation*}
	properties we denote respectively by $\Lambda$ being {\em symmetric} and {\em invariant under $\mathcal{G}$}.
	This is surely the case of the moment matrices after symmetrisation under the Reynolds operator (while some properties discussed here apply to non-symmetric matrices as well, our semidefinite programs and our numerical decomposition algorithm both employ symmetric matrices only).
	In the isotypic basis, we decompose $\tilde{\Lambda}_\text{iso} =  U^\top_\text{iso} \Lambda U_\text{iso}$ into blocks, each block corresponding to an isotypic subspace $W^r$:
	\begin{equation}
	\label{Eq:IsotypicBlockDiagonal}
	\tilde{\Lambda}_\text{iso} = U_\text{iso}^\top ~ \Lambda ~ U_\text{iso} =
	\begin{pmatrix}
	\tilde{\Lambda}^{1}_\text{iso} & 0 & \ldots & 0 \\
	0 & \tilde{\Lambda}^{2}_\text{iso} &        & 0 \\
	\ldots       &              &        & \ldots       \\
	0 & 0 &        & \tilde{\Lambda}^{R}_\text{iso}
	\end{pmatrix}\;,
	\end{equation}
	where $\tilde{\Lambda}^{r}_\text{iso} \in \mathbb{R}^{\dim W^r \times \dim W^r}$ and the off-diagonal blocks are zero by Schur's lemma.
	Each diagonal block $\tilde{\Lambda}^{r}_\text{iso}$ satisfies the invariance condition:
	\begin{equation}
	\label{Eq:LambdaTransforms}
	\tilde{\Lambda}^{r}_\text{iso} = (\tilde{M}_{\pi \text{iso}}^r)^\top ~ \tilde{\Lambda}^{r}_\text{iso} ~  \tilde{M}_{\pi  \text{iso}}^{r}\;.
	\end{equation}
	
	Now, let $G \in \mathbb{R}^{|\mathcal{J}| \times |\mathcal{J}|}$ be any symmetric real matrix and $\tilde{G}_\text{iso} = U_\text{iso}^\top ~ \Gamma ~ U_\text{iso}$ its form in the isotypic basis.
	We split $\tilde{G}_\text{iso}$ into blocks $\tilde{G}_\text{iso}^{i,j}$ according to the isotypic subspaces; as $\tilde{G}_\text{iso}$ is not invariant under the action of $\mathcal{G}$, its off-diagonal blocks are not necessarily zero.
	We now assume that $\Lambda$ comes from the projection of $G$ into the invariant subspace by the Reynolds operator of section~\ref{Sec:SpeedingUp}, $\Lambda = \mathcal{R}_\mathcal{G} (G)$.
	In the isotypic basis, we have:
	\begin{equation}
	\label{Eq:SymmetrisationBlock}
	\tilde{\Lambda}^{r}_\text{iso} = \frac{1}{|\mathcal{G}|} \sum_{\pi \in \mathcal{G}} (\tilde{M}^r_{\pi \text{iso}})^\top \tilde{G}^{r,r}_\text{iso} \tilde{M}^r_{\pi \text{iso}} \end{equation}
	Note that the form~\eqref{Eq:IsotypicBlockDiagonal} leads to efficient tests of semidefinite positiveness: the condition $\Lambda \geq 0$ is equivalent to $\tilde{\Lambda}_\text{iso} \geq 0$, which is efficiently written $\tilde{\Lambda}_\text{iso}^r \geq 0$ for all $r$.
	
	\subsubsection{Irreducible decomposition}
	The isotypic decomposition can be further refined.
	We can require of a change of basis $U_\text{irr}$ to decompose the permutation matrices $M_\pi$ as
	\begin{equation}
	\label{blkdiag_grp}
	\tilde{M} = U_\text{irr}^\top M_\pi U_\text{irr} = \underbrace{\tilde{M}_{\pi,1}^1 \boxplus \ldots \boxplus \tilde{M}_{\pi,m_1}^1}_{\tilde{M}_{\pi \text{iso}}^1} \boxplus \ldots \boxplus \underbrace{\tilde{M}_{\pi,1}^R \boxplus \ldots \boxplus \tilde{M}_{\pi,m_R}^R}_{\tilde{M}_{\pi \text{iso}}^R}\;,
	\end{equation}
	where, for each $r$, the $\{\tilde{M}_{\pi,i}^1\}_i$ express an irreducible representation of $\mathcal{G}$; the block matrices of the same irreducible representation are equivalent up to a similarity transformation (more on that below).
	Accordingly, the space $V$ splits each isotypic component $W^r$ into $m_r$ irreducible components:
	\begin{equation}
	\label{blkdiag_space}
	V = \underbrace{\big( V^{1}_1 \oplus \ldots \oplus V^{1}_{m_1} \big)}_{W^1} \oplus \ldots \oplus \underbrace{\big( V^{R}_{1} \oplus \ldots \oplus V^{R}_{m_R} \big)}_{W^R}\;,
	\end{equation}
	where $m_r$ is the multiplicity of the $r$-th irreducible representation and $d_r = \dim V^r_i$ its dimension.
	For each $r=1,\ldots,R$ and $i=1,\ldots,m_r$, we write $\{\vec{v}^r_{i,1}, \ldots, \vec{v}^r_{i,d_r}\}$ the basis vectors of $V^r_i$, which form the columns of the change of basis matrix $U_\text{irr} = \left( \vec{v}^{1}_{1,1}, \ldots, \vec{v}^{R}_{m_R, d_R} \right)$.
	The irreducible decomposition is stricter than the isotypic decomposition: each $U_\text{irr}$ provides a valid isotypic decomposition $U_\text{iso}$, but the converse is not true.
	The decomposition~\eqref{blkdiag_grp} is defined up to a change of basis in each component.
	For arbitrary orthonormal matrices $Y^r_i$, the following transformation 
	\begin{equation*}
	U_\text{irr}' = U_\text{irr} \big( \underbrace{Y_1^1 \boxplus \ldots \boxplus Y_{m_1}^1}_{\text{for } W^1} \boxplus \ldots \boxplus \underbrace{Y_1^R \boxplus \ldots \boxplus Y_{m_R}^R}_{\text{for } W^R} \big)
	\end{equation*}
	provides another orthonormal change of basis matrix that preserves the decomposition~\eqref{blkdiag_grp}.
	We can remove some degeneracy by picking, for each representation, matrices $\{Y_2^r, \ldots Y_{m_r}^r\}$ so that all $\tilde{M}^r_{\pi,i}$ have the same form $\tilde{M}^r_{\pi,i} = \tilde{M}^r_{\pi}$.
	We write $U$ a change of basis matrix that has the property
	\begin{equation}
	\label{blkdiag_grp1}
	\tilde{M} = U^\top M_\pi U = \underbrace{\tilde{M}_\pi^1 \boxplus \ldots \boxplus \tilde{M}_\pi^1}_{m_1 \text{ times} = \mathbbm{1}_{m_1} \otimes \tilde{M}^1_\pi} \boxplus \ldots \boxplus \underbrace{\tilde{M}_\pi^R \boxplus \ldots \boxplus \tilde{M}_\pi^R}_{m_R \text{ times} = \mathbbm{1}_{m_R} \otimes \tilde{M}^R_\pi}\;,
	\end{equation}
	where $\otimes$ is the Kronecker product (with the convention that $\mathbbm{1} \otimes X = X \boxplus \ldots \boxplus X$) and $\tilde{M}_\pi^r\in\mathbb{R}^{d_r \times d_r}$ corresponds to the blocks of $\tilde{M}_\pi$.
	This block-diagonal form of $\tilde{M}_\pi$ highlights again the invariant subspaces of $V$.
	Note that a finite group $\mathcal{G}$ has a finite number of irreducible linear representations over the reals.
	The question we will solve later is to identify which representations are present in $M_\pi$ and compute the change of basis matrix $U$.
	
	\subsubsection{Irreducible decomposition: impact on invariant symmetric matrices}
	As $U$ is a valid change of basis matrices for the isotypic decomposition, any symmetric invariant matrix $\Lambda$ still has the block diagonal form~\eqref{Eq:IsotypicBlockDiagonal}.
	Moreover, each isotypic block satisfies the invariance condition:
	\begin{equation*}
	\tilde{\Lambda} = U^\top ~ \Lambda ~ U = \tilde{\Lambda}^1 \boxplus \ldots \boxplus \tilde{\Lambda}^R, \qquad  \tilde{\Lambda}^r =  (\mathbbm{1}_{m_r} \otimes \tilde{M}^r_\pi)^\top ~ \tilde{\Lambda}^r ~ (\mathbbm{1}_{m_r} \otimes \tilde{M}^r_\pi), \quad \forall \pi \in \mathcal{G} \;.
	\end{equation*}
	Depending on the type of the representation $\tilde{M}^r_\pi$, the block $\tilde{\Lambda}^{r}$ will take different forms (see~\cite[13.2]{Serre1977}).
	For simplicity, we restrict our discussion to irreducible representations of real type.
	Irreducible representations are always of real type when $\mathcal{G}$ is {\em ambivalent}~\cite{footnote4,Armeanu}.
	Ambivalent groups include symmetric groups, dihedral groups and their direct products.
	Extensions of the technique and precision improvements will be presented in a future work~\cite{Denis}.
	For representations of real type, all blocks have the form $\tilde{\Lambda}^{r} = L^r \otimes \mathbbm{1}_{d_r}$ for a symmetric matrix $L^r \in \mathbb{R}^{m_r \times m_r}$:
	\begin{equation}
	\label{Eq:DecFine}
	\tilde{\Lambda}^{r} = L^r \otimes \mathbbm{1}_{d_r} =
	\begin{pmatrix}
	L^r_{1,1} \mathbbm{1}_{d_r} & L^r_{1,2} \mathbbm{1}_{d_r} & \ldots & L^r_{1,m_r} \mathbbm{1}_{d_r} \\
	L^r_{2,1} \mathbbm{1}_{d_r} & L^r_{2,2} \mathbbm{1}_{d_r} & \ldots & L^r_{2,m_r} \mathbbm{1}_{d_r} \\
	\ldots                &                       &        & \ldots \\
	L^r_{m_r, 1} \mathbbm{1}_{d_r} & L^r_{m_r,2} \mathbbm{1}_{d} & \ldots & L^r_{m_r,m_r} \mathbbm{1}_{d_r}
	\end{pmatrix}\;,
	\end{equation}
	where the $L^r$ do not have any restrictions beyond $(L^r)^\top = L^r$.
	The form~\eqref{Eq:DecFine} leads to further efficiency gains.
	The condition $\Lambda \geq 0$ is equivalent to $L^r \geq 0$ for all $r$, as $\mathbbm{1}_{d_r} \otimes L^r$ and $L^r \otimes \mathbbm{1}_{d_r}$ have the same eigenvalues (in fact, the difference between $\mathbbm{1}_{d_r} \otimes L^r$ and $L^r \otimes \mathbbm{1}_{d_r}$ is just a matter of convention in the enumeration of the basis vectors).
	
	Hence, the structure revealed by real linear representation theory of finite groups can be summed up by the following three equations:
	\begin{alignat}{2}
	V                     &= (\mathds{R}^{m_1}\otimes V^1)                 & \oplus \ldots \oplus & (\mathds{R}^{m_R}\otimes V^R) \;,     \label{blkdiag_space_abstract}\\
	U^\dagger ~M~ U=\tilde{M} &= (\mathbbm{1}_{m_1} \otimes \tilde{M}^1_\pi) & \boxplus \ldots \boxplus & (\mathbbm{1}_{m_R} \otimes \tilde{M}^R_\pi) \;,  \\
	U^\dagger~ \Lambda~ U=\tilde{\Lambda}&= (L^1 \otimes \mathbbm{1}_{d_1}) & \boxplus \ldots \boxplus & (L^R \otimes \mathbbm{1}_{d_R}) \;,
	\end{alignat}
	where all $V^r_i$ are isomorphic to $V^r$.
	
	Given an arbitrary symmetric matrix $G$, we obtain the symmetrised $\Lambda = \mathcal{R}_\mathcal{G}(G)$ by computing the Reynolds operator in two ways.
	First, we can apply the averaging sum described in section~\ref{Sec:SpeedingUp}.
	An efficient method is to take advantage of the form~\eqref{Eq:DecFine}.
	As the change of basis matrix is orthonormal, the projection to the symmetric subspace is orthogonal as well.
	Thus the coefficients of the blocks $L^r$ can be computed simply by averaging over the diagonal elements of each block in~\eqref{Eq:DecFine}:
	\begin{equation}
	\label{Eq:AverageBlock}
	L^r_{ij} = \frac{1}{d_r} \sum_k ~ (\vec{v}^r_{i,k})^\top ~ \Gamma ~ \vec{v}^r_{j,k} \;.
	\end{equation}
	
	\subsection{Symmetrisation exploiting  block-diagonalisation}
	We now describe step-by-step the construction of the three variants {\ttfamily isotypic}, {\ttfamily irreps} and {\ttfamily blocks} exploiting block-diagonalisation.
	
	\subsubsection{Partial block-diagonalisation: {\ttfamily isotypic}}
	
	We first work at the level of the isotypic subspaces $\{ W^r \}$ to provide a partial block-diagonalisation of the problem.
	We now present a simple recipe to discover the basis $U_\text{iso}$, inspired by~\cite{blkdiag, Maehara2010}.
	First, we obtain a generic random matrix $\Lambda$ satisfying the conditions~\eqref{Eq:CondLambda}.
	The procedure below requires $\Lambda$ to have well separated eigenvalues in a yet unknown basis (note that sampling such matrices from moment matrices would not work, as moment matrices often have additional structure).
	Thus, we sample a random symmetric matrix $G$ from the {\em Gaussian Orthogonal Ensemble} (GOE)~\cite{Anderson2009}, which are matrices whose entries are independently sampled from the normal distribution.
	Such matrices have well-separated, independently distributed eigenvalues whose distribution does not depend on a particular choice of basis.
	We obtain the desired matrix by symmetrising $\Lambda = \mathcal{R}_\mathcal{G}(G)$ according to the optimised Reynolds operator of section~\ref{Sec:SpeedingUp}.
	The following proposition will help us identify the isotypic basis $U_\text{iso}$.
	\begin{proposition}
		Let $\Lambda$ be a generic symmetric invariant matrix obtained by sampling from the GOE and applying the Reynolds operator.
		Generically, each eigenspace of $\Lambda$ is contained within a single isotypic subspace $W^i$.
	\end{proposition}
	\begin{proof}
		For the proposition to be true, we need to show that eigenvalues are not repeated across isotypic subspaces and that possible multiplicities only occur within an isotypic component.
		Recall that $\tilde{G}_\text{iso}$ and $\tilde{\Lambda}_\text{iso}$ have the form
		\begin{equation*}
		\tilde{G}_\text{iso} =  U_\text{iso}^\top G U_\text{iso} =
		\begin{pmatrix}
		\tilde{G}_\text{iso}^{1,1} & \ldots & \tilde{G}_\text{iso}^{1,R} \\
		\ldots & & \ldots \\
		\tilde{G}_\text{iso}^{R,1} & \ldots & \tilde{G}_\text{iso}^{R,R}
		\end{pmatrix}, \qquad \tilde{\Lambda}_\text{iso} = U_\text{iso}^\top \Lambda U_\text{iso} =
		\begin{pmatrix}
		\tilde{\Lambda}_\text{iso}^{1} & \ldots & 0 \\
		\ldots & & \ldots \\
		0 & \ldots & \tilde{\Lambda}_\text{iso}^{R}
		\end{pmatrix}
		= \tilde{\Lambda}_\text{iso}^{1} \boxplus \ldots \boxplus \tilde{\Lambda}_\text{iso}^{R}\;,
		\end{equation*}
		and $\tilde{G}_\text{iso}^{r,r}$ are submatrices of a matrix sampled from the GOE and thus have independent, random and well separated eigenvalues.
		Note that the block $\tilde{\Lambda}_\text{iso}^r$ is obtained by symmetrising the corresponding block $\tilde{G}_\text{iso}^{r,r}$ by~\eqref{Eq:SymmetrisationBlock} and only that block.
		The resulting symmetrised blocks $\tilde{\Lambda}^{i}$ will see their eigenvalue distribution modified.
		However, eigenvalues are still distributed independently {\em between} blocks and thus different blocks cannot share the same eigenvalue, as this happens almost never.
		Thus, the eigenspaces of $\tilde{\Lambda}$ do not overlap the block boundaries.
	\end{proof}
	
	As the isotypic subspaces $W^i$ are composed of eigenspaces of $\tilde{\Lambda}_\text{iso}$, which are also the eigenspaces of $\Lambda$ itself, the unordered vectors composing the change of basis matrix $U_\text{iso}$ are obtained simply from the eigenvalue decomposition of $\Lambda = T D T^\top$, where $T^{-1} = T^\top$ and $D$ is diagonal.
	However, this decomposition does not identifies which eigenspaces belong to the same isotypic component.
	For that, it is sufficient to sample a {\em second} symmetric invariant matrix $\Lambda'$, compute $T^\top \Lambda' T$ and find the reordering of columns of $T$ that brings $\Lambda'$ into its block-diagonal form.
	As, generically, all off-diagonal blocks $\tilde{\Lambda}'^{i,j}_\text{iso}$ will be zero (and only those), this identifies the requested change of basis $U_\text{iso}$.
	
	After having obtained the change of basis matrix $U_\text{iso}$, we proceed as follows to sample the basis in the {\ttfamily isotypic} method.
	As in Algorithm~\ref{algSampling}, we compute at every step $\ell$ a symmetrised sample $\Gamma'$.
	However, we do not directly store $\Gamma'$ as a basis element.
	Rather, we compute $\tilde{\Gamma}'_\text{iso} = U_\text{iso}^\top ~ \Gamma' ~ U_\text{iso}$, which is block diagonal with blocks $\tilde{\Gamma}^{r}$ according to~\eqref{Eq:IsotypicBlockDiagonal}, and only store the resulting blocks.

	\subsubsection{Fine block-diagonalisation: finding the irreducible basis}
	\label{Sec:FineBlockDiagonalization}
	
	We now move to complete block-diagonalisation.
	We assume we already identified the isotypic components and know that we need to adjust the bases of the $r$-th isotypic component $W^r$ using a change of basis matrix $U_r$ to obtain the full change of basis matrix $U$:
	\begin{equation*}
	U = U_\text{iso} \big( U^1 \boxplus U^2 \boxplus \ldots \boxplus U^R \big )\;,
	\end{equation*}
	so that $U^\top M_\pi U$ is fully block-diagonal according to~\eqref{blkdiag_grp1}.
	Let us revisit the symmetrised sample $\Lambda$, which we transform in the isotypic basis:
	\begin{equation*}
	U_\text{iso}^\top ~ \Lambda ~ U_\text{iso} = \tilde{\Lambda}^1_\text{iso} \boxplus \ldots \tilde{\Lambda}^R_\text{iso}\;.
	\end{equation*}
	We are looking for change of basis matrices $\{U^r\}$, inside each isotypic component, such that the $r$-th block $(U^r)^\top \tilde{\Lambda}^r_\text{iso} U^r = \tilde{\Lambda}^{r}$ satisfies~\eqref{Eq:LambdaTransforms} and $\tilde{\Lambda}^r$ has the form~\eqref{Eq:DecFine}.
	We treat all isotypic components separately.
	For simplicity, we now focus on the first block $r=1$ and write $m = m_1$, $d = d_1$, $L = L^1$. Remember~\eqref{Eq:DecFine}:
	\begin{equation*}
	\tilde{\Lambda}^1 = L \otimes \mathbbm{1}_{d} =
	\begin{pmatrix}
	L_{11} \mathbbm{1}_{d} & L_{12} \mathbbm{1}_{d} & \ldots & L_{1m} \mathbbm{1}_{d} \\
	L_{21} \mathbbm{1}_{d} & L_{22} \mathbbm{1}_{d} & \ldots & L_{2m} \mathbbm{1}_{d} \\
	\ldots                &                       &        & \ldots \\
	L_{m1} \mathbbm{1}_{d} & L_{m2} \mathbbm{1}_{d} & \ldots & L_{mm} \mathbbm{1}_{d}
	\end{pmatrix} , \qquad L \in \mathbb{R}^{m \times m}.
	\end{equation*}
	
	We now use the properties of this form to discover the change of basis matrix from samples of the isotypic component $\tilde{\Lambda}^1_\text{iso}$.
	Let $L = T D T^\top$ be the eigenvalue decomposition of $L$, where $D= \text{diag}(\lambda_1, \ldots, \lambda_m)$.
	We directly obtain the eigenvalue decomposition of $\tilde{\Lambda}^1$ by writing $\tilde{\Lambda}^1 = (T \otimes \mathbbm{1}_d) (D \otimes \mathbbm{1}_d) (T^\top \otimes \mathbbm{1}_d)$.
	As $L$ comes originally from a generic sample and was then symmetrised using~\eqref{Eq:SymmetrisationBlock}, its eigenvalues are each repeated $d$ times but are otherwise distinct.
	As eigenvalues do not depend on a choice of basis, we can exploit that property.
	
	Given $\tilde{\Lambda}^1$, what is the family of bases in which it is diagonal?
	As $D \otimes \mathbbm{1}_d =   (T \otimes \mathbbm{1}_d) ~ \tilde{\Lambda}^1 ~ (T^\top \otimes \mathbbm{1}_d)$, one possible change of basis matrix is $(T \otimes \mathbbm{1}_d)$. However, remark that
	\begin{equation*}
	D \otimes \mathbbm{1}_d = \begin{pmatrix} \lambda_1 \mathbbm{1}_d & &  \\ & \ldots & \\ & & \lambda_m \mathbbm{1}_d \end{pmatrix}
	\qquad = \qquad
	\begin{pmatrix} Y_1^\top & &  \\ & \ldots & \\ & & Y_m^\top \end{pmatrix} 
	\begin{pmatrix} \lambda_1 \mathbbm{1}_d & &  \\ & \ldots & \\ & & \lambda_m \mathbbm{1}_d \end{pmatrix}
	\underbrace{\begin{pmatrix} Y_1  & &  \\ & \ldots & \\ & & Y_m \end{pmatrix}}_{Y}\;,
	\end{equation*}
	where $Y_i$ are arbitrary orthonormal matrices.
	Hence, the full class of solution are the $\{(T \otimes \mathbbm{1}_d) Y\}$, where $Y = Y_1 \boxplus Y_2 \boxplus \ldots \boxplus Y_m$ and the $Y_i$ are orthonormal matrices.
	
	Hence we can proceed as follows.
	Having obtained the isotypic change of basis $U_\text{iso}$ using the method of the previous section, we consider a first sample of the current isotypic component $\tilde{\Lambda}^1_\text{iso}$.
	We compute its eigendecomposition $P^\top ~ \tilde{\Lambda}^1_\text{iso} ~ P = D\otimes \id_d$. As we characterised the family of bases in which $\tilde{\Lambda}^1_\text{iso}$ is diagonal, we have the guarantee that
	\[
	P = U^1 ~ Y \qquad \text{ with } \qquad Y = Y_1 \boxplus Y_2 \boxplus \ldots \boxplus Y_m\;,
	\]
	where $U^1$ is the change of basis matrix we are looking for and the eigendecomposition algorithm will return a random choice for $Y$.
	We then then obtain a second sample $\hat{\Lambda}^1_\text{iso}$ of the current isotypic component and change its basis using $P$ (note the use of $\hat{\cdot}$ instead of $\tilde{\cdot}$).
	Due to the presence of $Y$ we obtain:
	\begin{equation}
	\label{Eq:NearlyBasis}
	P^\top ~ \hat{\Lambda}^1_\text{iso} ~ P = Y^\top ~ \underbrace{(U^1)^\top ~ \hat{\Lambda}^1_\text{iso} ~ U^1}_{\text{in the form~}\eqref{Eq:DecFine}} ~ Y = 
	\begin{pmatrix}
	\hat{L}_{11} (Y_1^\top Y_1) & \hat{L}_{12} (Y_1^\top Y_2) & \ldots & \hat{L}_{1m} (Y_1^\top Y_m) \\
	\hat{L}_{21} (Y_2^\top Y_1) & \hat{L}_{22} (Y_2^\top Y_2) & \ldots & \hat{L}_{2m} (Y_2^\top Y_m) \\
	\ldots                &                       &        & \ldots \\
	\hat{L}_{m1} (Y_m^\top Y_1) & \hat{L}_{m2} (Y_m^\top Y_2) & \ldots & \hat{L}_{mm} (Y_m^\top Y_m)
	\end{pmatrix} .
	\end{equation}
	For invariant matrices $(Y_1^\top \mathbbm{1}_d Y_1) = \mathbbm{1}_d$, thus the choice of $Y_1$ does not impact the form~\eqref{Eq:DecFine}: it will however change the matrices of the irreducible representation $\tilde{M}_\pi^1$, corresponding to the arbitrariness in the choice of its basis.
	Now, we force all copies to be expressed in the same basis by multiplying the matrix $P$ with a correction factor, which provides the desired $U^1$:
	\[
	U^1 = P ~ \Big ( \mathbbm{1} \boxplus (Y_2^\top ~ Y_1) \boxplus \ldots \boxplus (Y_m^\top ~ Y_1) \Big)\;,
	\]
	and by looking at the first row of blocks in the matrix $P^\top \hat{\Lambda}_\text{iso}^1 P$, we directly have access to $(Y_i^\top Y_1)$, up to a constant factor $\hat{L}_{i1}$ which is easily corrected, as $(Y_i^\top Y_1)$ is orthonormal.
	
	\subsubsection{Fine block-diagonalisation: {\ttfamily irreps}}
	
	Given a irreducible change of basis $U$, for the {\ttfamily irreps} method we perform our processing of the samples as follows.
	As in Algorithm~\ref{algSampling}, we compute at every step $\ell$ a symmetrised sample $\Gamma'$.
	However, we do not directly store $\Gamma'$ as a basis element.
	Rather, we compute $\tilde{\Gamma} = U^\top ~ \Gamma' ~ U$, which is block diagonal with blocks $\tilde{\Gamma}^{r}$, each of the form $\tilde{\Gamma}^{r} = L^r \otimes \mathbbm{1}_{d_r}$ according to~\eqref{Eq:DecFine}.
	Instead of taking an arbitrary copy of $L^r$ in the matrix, we get the resulting block from the average of all copies of $L^r$ present.
	As we no longer need to store multiple copies of the same block and can safely discard off-diagonal elements, the storage and computational requirements for the basis construction are dramatically decreased.
	
	\subsubsection{Sampling directly the blocks: {\ttfamily blocks}}
	
	Another technique is to sample directly from the blocks, bypassing the explicit evaluation of the Reynolds operator as in Section~\ref{Sec:SpeedingUp}.
	Let us compute the moment matrix $\Gamma$ directly in the block-diagonal basis, using the pre-computed $\hat{s}_{\alpha}$ of Algorithm~\ref{algSampling}:
	\begin{equation*}
	U^\top ~ \Gamma ~ U = \sum_{\alpha \beta} U_{\alpha j} \tr [\hat{s}_{\alpha}^{\dag}  \hat{s}_{\beta}] U_{\beta k} \;.
	\end{equation*}
	We pre-compute  $\omega_j = \sum_{\alpha} U_{\alpha j}  \hat{s}_{\alpha} = \sum_{\alpha} U_{\alpha j} s_j(\overline{X}) \overline{K}$, so that the element $(U^\top ~ \Gamma ~ U)_{j,k}$ is computed without much effort:
	\begin{equation*} (U^\top ~ \Gamma ~ U)_{j,k} = \tr [\omega_j^{\dag} \omega_k] \;. \end{equation*}
	
	Remember that $\Gamma$ has not been through the explicit Reynolds operator and is not invariant under $\mathcal{G}$.
	However, we can use the fast projection~\eqref{Eq:AverageBlock} and compute only the coefficients that are required without forming the complete moment matrix.
	We then proceed as with {\ttfamily irreps} to construct the symmetrised basis by storing the blocks $L^r$.
	
	\subsubsection{Impact of the methods on the RAC for $n = 2$ and $d = 3$}
	
	We consider the RAC example presented in Table~\ref{tabImplementations} for $d = 3$.
	For the choice of monomials corresponding to $\mathbbm{1}$, $\rho_x$, $M^b_y$ and $\rho_x M^b_y$, we obtain a generating set of size 70; thus, without block diagonalisation, the moment matrix has size $70 \times 70$.
	Without symmetrisation, the number of samples is $545$.
	The symmetry group has order $72$. Applying averaging under the Reynolds operator ({\ttfamily reynolds}) reduces the number of samples to $13$; this number of samples will not be reduced further, however the moment matrix can be block diagonalised.
	Applying the {\ttfamily isotypic} block diagonalisation, we identify blocks of size $2$, $3$, $4$, $5$, $12$, $16$ and $28$.
	Refining further ({\ttfamily irreps} or {\ttfamily blocks}), we split those blocks further and obtain a final block decomposition of sizes $1$, $1$, $3$, $3$, $4$, $5$ and $7$.
	As we see in the next section, both the number of samples and the block sizes of the finest decomposition do not depend on $d$.


		\section{Application to random access coding}\label{RACs}
		We exemplify the general symmetrisation technique by considering a generalisation to many inputs of the symmetrisation proposed in~\cite{Sym8} of the two-party computation task known as a random access code (RAC) \cite{Ambainis, TavakoliRACs}. In a RAC, a party Alice receives random inputs $x=x_1,\ldots,x_{n}\in [d]$, and another party Bob receives a random input $y\in[n]$. By receiving a $d$-dimensional quantum system $\rho_{x}$ from Alice, Bob measures $\{M_y^b\}_b$ with outcome $b\in[d]$, aiming to recover Alice's $y$'th input. The average  success probability is
		\begin{equation}\label{racscore}
		\mathcal{A}_{n,d}^{\text{RAC}}=\frac{1}{nd^n}\sum_{x,y}\Tr(\rho_{x} M_y^{x_y}).
		\end{equation}
		We apply a symmetrised semidefinite relaxation as described by the general recipe to upper bound $\mathcal{A}_{n,d}^{\text{RAC}}$ for any states and rank-one projective measurements. To this end, we first identify  generators of the symmetry group, i.e., the re-labellings of inputs/outputs of Alice and Bob that leave the problem invariant.   Due to the  simplicity of the objective function, the symmetries can be spotted by direct inspection.

		\begin{table}
			\centering
			\begin{tabular}{|c|c|c|c|c|c|}
				\hline
				& \multicolumn{2}{|c|}{$\#$ Basis elements} & \multicolumn{2}{|c|}{SDP (+ blkdiag) time (sec) } & \\
				\hline
				$(n,d)$ & \text{ standard } &  \text{sym} & \text{standard} & \text{sym} & Result  \\ [0.5ex]
				\hline
				(3,2) & 224  &  28 &  11  & 2  & 0.7887  \\
				(3,3) & 11380  & 82  & $>8.5\times 10^4$ & 4 & 0.6989  \\
				(3,4) & -  & 82 & - & 15  & 0.6474\\
				(3,5) & -  &  82  & - & 120 & 0.6131\\
				\hline
			\end{tabular}
			\caption{Comparison between symmetrised and standard implementation for RACs. The symbol "$-$" indicates that we were unable to perform a computation. Note that the reduction in the number of basis elements leads to an analogous reduction in the sampling time. 
			}\label{tablerac1}
		\end{table}

		We identify $n+1$ types of generators. In the following, $S_n$ denotes the symmetric group of degree $n$. The first type $\pmb{\xi}$ is parameterised by $\xi\in S_n$ and corresponds to a permutation of the indices in the input string $x_1,\ldots,x_n$, while correcting $y$. The remaining $n$ types $\pmb{\pi_1},\ldots,\pmb{\pi_n}$ are parameterised by  permutations $\pi_1,\ldots ,\pi_n\in S_d$   of the $d$ possible values of $x_1,\ldots,x_n$ respectively, while correcting $b$. Specifically,
		\begin{align}
		\label{eqRACSymmetries}\nonumber
		\pmb{\xi}(\rho_{x_1,\ldots, x_n}) &= \rho_{x_{\xi(1)},\ldots, x_{\xi(n)}},& \qquad \pmb{\xi}(M_y^b) &= M_{\xi(y)}^b,&\\\nonumber
		\pmb{\pi_1}(\rho_{x_1,\ldots, x_n}) &= \rho_{\pi_1(x_1), x_2,\ldots, x_n},& \qquad \pmb{\pi_1}(M_1^b) &= M_1^{\pi_1(b)},&\\\nonumber
		\vdots\\
		\pmb{\pi_n}(\rho_{x_1,\ldots, x_n}) &= \rho_{x_1,x_2,\ldots, \pi_n(x_n)},& \qquad \pmb{\pi_n}(M_n^b) &= M_n^{\pi_n(b)},&
		\end{align}
		and $\pmb{\pi_k}$ leaves  $M_l^b$  unaffected for $k\neq l$. By simple enumeration, we observe that any element in $\pi \in \mathcal{G}$, for given $d$, can be written as the composition of $n+1$ transformations $\pi = \pmb{\xi} \pmb{\pi_1}\ldots\pmb{\pi_n}$. These transformations  are compatible with the structure of the problem and leave the average success probability $\mathcal{A}_{n,d}^{\text{RAC}}$ invariant.

		Using these generators we have implemented the symmetrised relaxation and numerically block-diagonalised the collection of sampled moment matrices. The maximal quantum value of $\mathcal{A}_{n,d}^{\text{RAC}}$ in the case of $n=2$ is analytically known~\cite{Sym8}. This was previously used in Section \ref{secsec} to verify the numerical precision of our methods. Here, we focus on $n=3$ for which no analogous analytical result is known when $d>2$.  We choose the hierarchy level corresponding to a moment matrix generated by the products $\{\openone,\rho_x, M_y^b, \rho_x M_y^b\}$. In Table \ref{tablerac1} we  compare the computational requirements of the symmetrised and standard implementations. We find a dramatic reduction in the size of the sampled basis and a highly efficient subsequent SDP which straightforwardly overcomes the limitations encountered in \cite{TavakoliRACs}. As an illustration of the usefulness of block-diagonalisation, for $(n,d)=(3,5)$ the moment matrix is of size 2241 but is effectively treated as seven non-trivial blocks of size at most 448.

		\section{Application to Bell-inequality-based communication complexity problem with illustration of how to automatise the search for symmetries}
		Most correlation games involve reasonably complicated objective functions which have significant non-obvious symmetries that cannot easily be found by direct inspection. Therefore, it is important to consider two questions. 
		\begin{enumerate}[I]
			\item How useful is symmetrisation  when only a small number of symmetries are discovered?
			\item  How does one find (non-obvious) symmetries of any objective function in a given physical scenario? 
		\end{enumerate}
		We consider these matters in a distributed computation task \cite{BZ02, magic7, DM18} based on facet Bell inequalities \cite{cglmp}.

		\begin{table}
			\centering
			\scalebox{0.82}{
				\begin{tabular}{|c|c|c|c|c|c|c|c|c|c|}
					\hline
					& \multicolumn{3}{|c|}{$\#$ Basis elements} & \multicolumn{3}{|c|}{SDP (+ blkdiag) time (sec) } &\\
					\hline
					$d$ & \text{Standard} & \text{'Obvious' sym}  & \text{Full sym} & \text{Standard} & \text{'Obvious' sym}   & \text{Full sym} &  \text{Result}  \\ [0.5ex]
					\hline
					3 & 329 & 111   & 36 & 8 & 3 & 0.3 & 0.7287 \\
					4 &  1154 & 290  & 84 & 160 & 5 &  0.5 & 0.7432 \\
					5 &  3002 &  602 & 171 & 2100  & 30 & 1 & 0.7569 \\
					6 & 6497 & 1085 & 297 & 17000 & 150  & 2.5 & 0.8000  \\
					7 & - & 1775 & 482 & - & 650  & 7  & 0.8333 \\	
					\hline
			\end{tabular}}
			\caption{Comparison between standard implementation for $\mathcal{A}_d^{\text{facet}}$ and its symmetrised implementation using both the obvious symmetry and the full symmetry group.}\label{tab:cglmp}
		\end{table}

		Alice and Bob take random inputs  $x\in[2]_0$, $x_0\in[d]_0$ and  $y\in[2]_0$ respectively, where $[s]_0=\{0,\ldots, s-1\}$. Alice sends a $d$-dimensional system $\rho_{x,x_0}$ to Bob which he measures with $\{M_y^b\}$, where $b\in[d]_0$. The objective of the task is 
		\begin{equation}
		\mathcal{A}_d^{\text{facet}}=\frac{1}{4d}\sum_{k=0}^{\lfloor \frac{d}{2}\rfloor-1}c_k\sum_{x_0,x,y}\Tr\left[\rho_{x,x_0}\left(M_y^{f_0}-M_y^{f_1}\right)\right],
		\end{equation}
		where $c_k=1-2k/(d-1)$ and $f_j=x_0-xy-(-1)^{x+y+j}(k+j)$, for $j\in\{0,1\}$. The computations are modulo $d$. 
		
		There is one easily spotted symmetry, namely jointly shifting the value of $x_0$ and $b$. We write this as $\pmb{\pi^c}(\rho_{x,x_0})=\rho_{x,x_0+c}$ and $\pmb{\pi^c}(M_{y}^{b})=M_{y}^{b+c}$ for some $c\in[d]_0$, parameterised by a cyclic permutation of $d$ elements $\pi^c$. Considering only this 'obvious' symmetry, we address question (I) by considering the hierarchy level corresponding to products of the form $\{\openone, \rho_{x,x_0},M_y^b,\rho_{x,x_0}M_y^b,M_y^bM_{y'}^{b'}\}$, choosing rank-one projectors and implementing the NV hierarchy both with and without symmetry exploitation. The results in Table \ref{tab:cglmp} show that even this small symmetry group allows one to reduce the computational requirements of the problem many times over. Nevertheless, the advantages are much smaller than what was obtained 	for the RACs in section~\ref{RACs}. Therefore, we turn to question (II) and search for non-obvious symmetries. Using the MATLAB package \cite{package}, we  enumerated the elements of the ambient group for small $d$ and discovered that $\mathcal{A}_d^{\text{facet}}$ has a symmetry group of order $4d$, to be compared with the previous cyclic group of order $d$. We then generalised that group construction for all $d$. The elements of the symmetry group are constructed by considering all combinations of products of $\pmb{\pi^c}$ with either the group identity, one of the two additional symmetries 
		\begin{align}\nonumber
		&\pmb{\phi}(\rho_{x,x_0})=\rho_{\bar{x},d-1-x_0} & \pmb{\phi}(M_y^b)= M_y^{d-1-y-b} \\
		& \pmb{\varphi}(\rho_{x,x_0})=\rho_{x,d-\bar{x}-x_0} & \pmb{\varphi}(M_y^b)=M_{\bar{y}}^{d-1-b},
		\end{align}
		or the product of these two additional symmetries, where the bar-sign denotes bitflip. Implementing the NV hierarchy using the  full symmetry group (see Table~\ref{tab:cglmp}), we greatly improve on the results obtained with the obvious cyclic symmetries and straightforwardly overcome the computational limitations of \cite{magic7}.

		\section{Application to the dimension bounded  $I_{3322}$-like Bell inequality}\label{Bell}
		We consider bounding finite-dimensional quantum correlations in a Bell inequality test. We consider a modified version of the $I_{3322}$ Bell inequality (studied without symmetries in~\cite{NV2}):
		\begin{align}
		\label{eq:i3322}
		I_{3322}^{(c)}=c\Big(\langle A_1B_3\rangle+\langle A_3B_1\rangle-\langle A_2B_3\rangle-\langle A_3B_2\rangle\Big) + \nonumber \\
		\langle A_1\rangle\!+\!\langle A_2\rangle\!+\!\langle B_1\rangle\!+\!\langle B_2\rangle-\Big\langle\! (A_1+A_2)(B_1+B_2)\!\Big\rangle,
		\end{align}
		where $A_x$ and $B_y$, for $x,y=1,2,3$ are projective measurements with eigenvalues $\pm 1$ (which are optimal for binary outcomes).
		The local bound reads $I_{3322}^{(c)}\le 4c$. For $c=1$, we recover the original $I_{3322}$ inequality \cite{Froissart, Sliwa, CG}.
		For any value of $c$, this inequality is symmetric under the permutation of parties, which we write $\pmb{p}$: $\pmb{p}(A_z) = B_z$ and $\pmb{p}(B_z) = A_z$ for $z=1,2,3$, and under the correlated re-labelling of inputs and outputs $\pmb{r}$: $\pmb{r}(A_1)=A_2$, $\pmb{r}(A_2)=A_1$, and $ \pmb{r}(B_3) = -B_3$, while $A_3$, $B_1$ and $B_2$ are unaffected.
		By repeated composition, we obtain the symmetry group $G = \left \{ \pmb{\text{id}}, \pmb{p}, \pmb{r}, \pmb{pr}, \pmb{rp}, \pmb{prp}, \pmb{rpr}, \pmb{prpr}  \right \}$.
		
		We compute the quantum bound of~\eqref{eq:i3322} when $c=1,3/2,2$ and the dimension is bounded by $d=2,3,4$. We construct the relaxation according to the hierarchy level 4, which corresponds to a moment matrix of size $244\times244$. The space of symmetric moment matrices can be block-diagonalised to yield six blocks of size at most 61. Thanks to symmetrisation, one can reduce the number of rank combinations for the measurement operators from the original $(d+1)^6$ by discarding redundant combinations (see section~\ref{appPackage}). For each case we sample the considered measurements and pure states $\psi$ and compute the moment matrix $ \Gamma_{j,k}=\Big < s_j(X)^\dagger~ s_k(X) \Big >$ with $\left< S \right> = \left< \psi \middle | S \middle | \psi \right >$, for a product of operators $S$. We present the results in Table~\ref{tab:i3322}. The advantages due to symmetrisation enables us to efficiently evaluate the large number SDPs in the high hierarchy level \cite{CommentI3322}.

		\begin{table}
			\centering
			\scalebox{0.92}{
				\begin{tabular}{|c|c|c|c|c|c|c|c|c|}
					\hline
					\multicolumn{2}{|c|}{}	& \multicolumn{2}{|c|}{$\#$ Basis elements} & \multicolumn{2}{|c|}{SDP (+ blkdiag) time (sec) } & \\
					\hline
					$c$ & $d$ & \text{Standard} & \text{Sym} & \text{Standard} & \text{Sym} & Result \\ [0.5ex]
					\hline
					1  &2        & 1771 & 240 & 500 & 2 & 5.000 000 \\
					1  &3        & 3292     & 496 & 2900      & 6   & 5.000 000  \\			
					1  &4        & 4492    & 594 &  3500    & 10    & 5.003 502  \\
					1  &$\infty$ &      &     &      &     & 5.003 502 \\
					\hline 
				\end{tabular}
				
			}
			\begin{tabular}{|c|c|l|l|c|c|l|}
				\cline{1-3} \cline{5-7}
				c   & d & Result    &  & c & d & Result    \\ \cline{1-3} \cline{5-7} 
				1.5 & 2 &     6.250 000      &  & 2 & 2 & 8.013 177 \\ \cline{1-3} \cline{5-7} 
				1.5 & 3 &    6.354 110       &  & 2 & 3 & 8.050 117 \\ \cline{1-3} \cline{5-7} 
				1.5 & 4 & 6.380 669 &  & 2 & 4 & 8.075 937 \\ \cline{1-3} \cline{5-7} 
				1.5 & $\infty$  & 6.380 669 &  & 2 & $\infty$  & 8.075 938 \\ \cline{1-3} \cline{5-7} 
			\end{tabular}
			\caption{Comparison between symmetrised and standard implementations for $I_{3322}^{(c)}$ and dimension $d$. For $d=\infty$ we use the results of \cite{Sym7}. The number of basis elements and the solver time are reported for projective measurements of rank $\left\lfloor d/2 \right\rfloor$.  
			}\label{tab:i3322}
		\end{table}

	\section{Symmetrisation in a multiparty distributed computation}\label{appEx}
	Both the NV hierarchy and the symmetrisation technique straightforwardly extend to multipartite systems. In particular, due to the rapidly increasing computational requirements associated to increasing the number of parties, the use of symmetrisation is typically even more critical in such scenarios. Here, we exemplify the straightforward manner in which symmetrisation extends to multiparty scenarios, by considering  a distributed computation involving communicating parties that perform local transformations on an incoming state.

	Consider an $n+2$ party distributed computation involving parties $A_0,\ldots, A_{n+1}$, arranged in a line. The first party, $A_0$, receives random inputs $x_0,x_1\in[d]_0$, while  $A_1,\ldots,A_n$ independently receive random inputs $y_k\in[d]_0$. Party $A_{n+1}$ receives random inputs $z\in[2]_0$, $t\equiv t_1\ldots t_n\in [2]_0$ and produces an output $a\in[d]_0$. For $k\in[n+1]_0$, $A_k$ may only send a $d$-dimensional system to $A_{k+1}$. The task is fulfilled if $a=x_z+t\cdot y \mod{d}$, where $y=y_1,\ldots,y_n$. Denoting by $\rho_{x_0,x_1}^{y}$  the state that given to $A_{n+1}$, the average success probability is
	\begin{equation}
	\mathcal{A}_{n,d}^{\text{multi}}=\frac{1}{d^{n+2}2^{n+1}}\sum_{x_0,x_1,y,z,t}\Tr\left(\rho_{x_0,x_1}^y M_{z,t}^{x_z+t\cdot y}\right).
	\end{equation}
	For simplicity, we limit the transformations of the parties $A_1,\ldots,A_n$ to unitaries, $U_{k,y_k}$, and write $\rho_{x_0,x_1}^y=(U_{n,y_n}\ldots U_{1,y_1})\rho_{x_0x_1}(U_{n,y_n}\ldots U_{1,y_1})^\dagger$. We focus on the case of $A_{n+1}$ performing rank-one projective measurements.  We find several types of symmetries. Firstly, one may permute the labels of the inputs of $A_0$ while also permuting $z$. Secondly, one may cyclically permute the input $x_0$ ($x_1$) of $A_0$ while also permuting $b$ only if $z=0$ ($z=1$). Thirdly, for each of the parties $A_1,\ldots,A_n$, one may cyclically permute $y_k$ while also permuting $b$ only if $t_k=1$. These can be written   
	\begin{align}\nonumber
	& \pmb{\xi}(\rho_{x_0,x_1}^y)=\rho_{x_{\xi(0)},x_{\xi(1)}}^{y} & \pmb{\xi}(M_{z,t}^b)=M_{\xi(z),t}^b\\ \nonumber
	& \pmb{\pi_k}(\rho_{x_0,x_1}^y)=\rho_{x_0,x_1}^{\pi_k\cdot y} & 
	\pmb{\pi_k}(M_{z,t}^b)=\begin{cases}
	M_{z,t}^b  & \text{if } t_k=0 \\
	M_{z,t}^{\pi_k(b)} & \text{if } t_k=1
	\end{cases}\\
	& \pmb{\pi^0}(\rho_{x_0,x_1}^y)=\rho_{\pi^0(x_0),x_1}^{y}  & 
	\pmb{\pi^0}(M_{z,t}^b)=\begin{cases}
	M_{z,t}^{\pi^0(b)} & \text{if } z=0\\\nonumber
	M_{z,t}^b  & \text{if } z=1
	\end{cases}\\
	& \pmb{\pi^1}(\rho_{x_0,x_1}^y)=\rho_{x_0,\pi^1(x_1)}^{y}  & 
	\pmb{\pi^1}(M_{z,t}^b)=\begin{cases}
	M_{z,t}^b  & \text{if } z=0\\
	M_{z,t}^{\pi^1(b)} & \text{if } z=1
	\end{cases}
	\end{align}
	where  $\xi\in S_2$, $\pi_k, \pi^0, \pi^1$ are cyclic permutations of $d$ objects and $\pi_k\cdot y=(y_1,\ldots, \pi_k(y_k),\ldots y_n)$. Note that we have omitted a small number of additional symmetries, for example applying a non-cyclic permutation to $x_0$ and then applying the same permutation to $b$ given that $\forall k: t_k=0$ and $z=0$. A similar non-cyclic permutation can be made for $x_1$ and $b$. For simplicity, in our numerical implementation for this example, we have not exploited such symmetries.

	\begin{table}[h]
		\centering
		\begin{tabular}{|c|c|c|c|c|c|c|c|}
			\hline
			& \multicolumn{2}{|c|}{$\#$ Basis elements} & \multicolumn{2}{|c|}{SDP (+ blkdiag) time (sec) } & \\
			\hline
			$(n+2,d)$ & \text{Standard} & \text{Sym} &  \text{Standard} & \text{Sym} & Result \\ [0.5ex]
			\hline
			(5,2) & 2543 &  72 &  500 & 2 & 0.6250\\
			(6,2) & 10791  &  157  & - & 3 & 0.5884  \\
			(7,2) & $>2.9 \times 10^4 $  &  330  &  - & 10  & 0.5625  \\		
			(3,3) & $>2\times 10^4$  &  22  &  - & 2  & 0.6667   \\			
			(5,4) & -  & 651  &  - & 500  &  0.4375  \\			
			\hline
		\end{tabular}
		\caption{Comparison between symmetrised and standard implementation for $\mathcal{A}_{n,d}^{\text{multi}}$. }\label{tab:multi}
	\end{table}

	We have implemented the semidefinite relaxation with and without symmetries when considering  operators products of the form $\{\openone,\rho_{x_0x_1}, (\Pi_k U_{k,y_k})\rho_{x_0x_1}(\Pi_k U_{k,y_k})^\dagger, M_{z,t}^a \}$, for all $k=1,\ldots,n$ (see Table~\ref{tab:multi}). The block-diagonalisation method employed was the simple heuristic described in the main text.  We observe that symmetrisation dramatically reduces the computational requirements and allows for straightforward evaluation for cases involving many parties for which a standard method is found impractical.

	\section{Optimal symmetrisation of random access codes via irreducible decompositions of the representation}\label{appIrreps}

		Although the numerical approach to symmetrisation based on sampling is both highly efficient and simple to implement for specific problems, it provides little insight into the underlying reasons for the results it produces. Relevant such questions include; why the sample space is of a particular dimension, or how to interpret the blocks of the diagonalised SDP matrix, or how these properties evolve for a family of correlation scenarios. In order to answer such questions, more must resort to the more technically demanding issue of considering the symmetrisation problems by analytical means.  As an illustration of the insights provided by  such an analytical approach to symmetrisation, we derive the decomposition of the action of $\mathcal{G}$ into irreducible representations for the example of  RACs in section~\ref{RACs} for $n = 2$ and arbitrary $d$.

	\subsection{Overview}

	We consider the problem of optimally symmetrising, by fully analytical means, the family of RACs for $n=2$ and arbitrary $d$ for a hierarchy level corresponding to the operator products of the form  $\{\openone, \rho_{x_1x_2},M_y^b,\rho_{x_1x_2}M_y^b\}$.
	Note that symmetrisation by averaging over the Reynolds operator (method {\ttfamily reynolds}) in this family of RACs was already considered by numerical means, for a somewhat lower hierarchy level, in \cite{Sym8}.
	Here, we analytically find the full decomposition in irreps of the form Eq.~(\ref{blkdiag_space_abstract}) for any $d$.
	Table~\ref{Table_irreps} shows that seven irreps of various multiplicities appear in the irreps decomposition of $V$ (remember that $V$ is the column space of the moment matrix):
	\begin{equation}
	V= (\mathds{R}^5\otimes T)\oplus (\mathds{R}^3\otimes S)\oplus(\mathds{R}^7\otimes \phi)\oplus(\mathds{R}^4\otimes \pi_+)\oplus(\mathds{R}^3\otimes \pi_-)\oplus\Lambda\oplus\Omega\oplus\lambda\oplus\omega.
	\end{equation}
	Hence, as given by Eq.~(\ref{Eq:AverageBlock}), $\mathcal{R}(\Gamma)$ is defined by seven matrices of dimension $m_i^2$, for a total dimension 112.
	This shows that the dimension of the feasible set can directly be reduced to $112$, independently of the dimension $d$.
	Hence, the sampling technique explores a space of at most dimension $112$.
	Under the assumption that the dimension found with sampling should not decrease with $d$, it shows that this dimension should be stationary after some particular dimension $d^*$.
	In practice, for $d$ from 3 to 10, we obtained a further reduction of the dimension from $112$ to $13$.
	We conjecture this stationary property for any $n$, i.e. that $d^*=3$. 
	Explicit fully-analytical block-diagonalisation may be be a useful approach to tackle this conjecture : analysing in which irreps those 13 degrees of liberties are used is left for future work.

	\subsection{Symmetry adapted basis for semidefinite relaxation of high-dimensional RACs with $n=2$ }\label{explicit_decompo_RAC}
	
	\begin{table}
		\centering
		\begin{tabular}{|c|c|c|c|c|c|c|c|c|c|}
			\hline
			Irreps label $i$ & $T$ & $S$  & $\phi$ & $\pi_+$ & $\pi_-$ & $\Lambda$ & $\Omega$ & $\lambda$ & $\omega$  \\
			\hline
			Dimension $d_i$ & $1$ & $1$ & $2(d-1)$ & $(d-1)^2$ & $(d-1)^2$ & $(d-1)(d-2)$ & $d(d-3)$ & $(d-1)^2(d-2)$ & $d(d-1)(d-3)$    \\
			\hline
			Multiplicity in $\{\id\}$ & 1 & 0 & 0 & 0 & 0 & 0 & 0 & 0 & 0  \\
			Multiplicity in $\{\rho_{x_1x_2}\}$ & 1 & 0 & 1 & 1 & 0 & 0 & 0 & 0 & 0  \\
			Multiplicity in $\{M_b^x\}$ & 1 & 1 & 1 & 0 & 0 & 0 & 0 & 0 & 0  \\
			Multiplicity in $\{\rho_{x_1x_2}M_b^y\}$ & 2 & 2 & 5 & 3 & 3 & 1 & 1 & 1 & 1  \\
			\hline
			Total Multiplicity $m_i$ & 5 & 3 & 7 & 4 & 3 & 1 & 1 & 1 & 1  \\
			
			\hline
		\end{tabular}
		\caption{Irreps appearing into the decomposition of $V$ with dimension and multiplicities, in the domain of moment matrix monomials associated to $\{\id\}, \{\rho_{x_1x_2}\}, \{M_b^y\}$ and $\{\rho_{x_1x_2}M_b^y\}$. See Section~\ref{explicit_decompo_RAC} for definitions of each of these irreps.}\label{Table_irreps}
	\end{table}
	
	The standard basis of the corresponding $V$ is given by four blocks.
	the first one is of dimension 1, corresponding to $\{\id\}$. The second is of dimension $d^2$ and canonical basis $e_{x_1}^{X_1}\otimes e_{x_2}^{X_2}$ corresponding to $\{\rho_{x_1,x_2}\}$ for ${1\leq x_1,x_2\leq d}$. The third is of dimension $2d$ and canonical basis $e_{y}^{Y}\otimes e_{b}^{B}$ corresponding to $\{\rho_{y,b}\}$ for $1\leq b\leq d, y=1,2$. The last one is of dimension $2d^3$ and canonical basis $e_{x_1}^{X_1}\otimes e_{x_2}^{X_2}\otimes e_{y}^{Y}\otimes e_{b}^{B}$ corresponding to $\{\rho_{x_1,x_2}M_b^y\}$ for ${1\leq x_1,x_2,b\leq d}, y=1,2$.

	Let $\delta_\pm^Y=\frac{e_1^Y\pm e_2^Y}{\sqrt{2}}$. We express the symmetry adapted in terms of some known irreps of the symmetric group $S_d$.
	$S_d$ has a natural action over $\mathds{C}^d$ by permuting its canonical basis elements $\{e_x\}_{1\leq x \leq d}$. It decomposes into the trivial irrep $t$ of dimension 1 generated by $\delta_+=\frac{1}{\sqrt{d}}(\sum_x e_x)$ and the standard representation $\phi_1$, orthogonal to it. As usual in decomposition into irreducible representation, only the vectorial space matters, the choice of basis is necessary for computations but is arbitrary. An orthonormal basis of $\phi_1$ can be taken as $\xi_1\propto e_1-e_2$, $\xi_2\propto e_1+e_2-2e_3$, ..., $\xi_{d-1}\propto e_1+...+e_{d-1}-(d-1)e_d$. 
	We also introduce the notation $\delta_{ij}=\frac{e_i-e_j}{\sqrt{2}}$ and 
	$\alpha_{ij}=\delta_{ij}\otimes\delta_{ij}$.
	
	In the following, the representation $\phi_1\otimes\phi_1$, generated by the $\xi_i\otimes\xi_j$ also appears. 
	Its irreps decomposition under the action of $S_d$ is $\phi_1\otimes\phi_1=\Lambda\phi_1 \oplus t \oplus \phi_1 \oplus \theta$ where:
	\begin{itemize}
		\item $\Lambda\phi_1$ is the alternating square of $\phi_1$ of basis $\beta_{k}\propto \xi_i\otimes\xi_j-\xi_j\otimes\xi_i$, where $k=(i,j)$ and $1\leq i<j\leq d-1$. 
		\item $t\oplus\phi_1$ is a copy of the natural representation embedded into $\phi_1\otimes\phi_1$, with a canonical basis $\tilde{e}_k\propto \sum_{i\neq k}\alpha_{ik} - \sum_{k\neq i<j\neq k}\alpha_{ij}+\frac{d-4}{d}\sum_{i<j}\alpha_{ij}$. Basis $\tilde\delta_+$ of $t$ and $\tilde\xi_1 ... \tilde\xi_{d-1}$ can be obtained from the $\tilde{e}_k$ with the same formal expressions as previously (just adding tildes). 
		\item $\theta$ is a last irreps of dimension $d(d-3)/2$. A basis $u_k$ can be obtained by orthogonality.
	\end{itemize}
	
	We now give the decomposition of the four blocks independently.
	\\
	\\
	\textit{(i)} \textbf{ Block} $\{\id\}$
	\\
	This gives a first trivial representation $T$
	\\
	\\
	\textit{(ii)}\textbf{ Block} $\{\rho_{x_1x_2}\}$
	\\
	It decomposes as $T\oplus\phi\oplus\pi_+$, with:
	\begin{align}
	& T \text{ of basis } \delta_+^{X_1}\otimes\delta_+^{X_2}, & \pi_+ \text{ of basis } \xi_i^{X_1}\otimes\xi_i^{X_2}, && \phi \text{ of basis } \xi_i^{X_1}\otimes\delta_+^{X_2}, \delta_+^{X_1}\otimes\xi_i^{X_2}.
	\end{align} 
	\\
	\textit{(iii)}\textbf{ Block} $\{M_b^y\}$
	\\
	It decomposes as $T\oplus\phi\oplus S$, with:
	\begin{align}
	& T \text{ of basis } \delta_+^{Y}\otimes\delta_+^{B}, & S \text{ of basis } \delta_-^{Y}\otimes\delta_+^{B} && \phi \text{ of basis } e_1^{Y}\otimes\xi_i^{B}, e_2^{Y}\otimes\xi_i^{B}.
	\end{align}
	\\
	\textit{(vi)}\textbf{ Block} $\{\rho_{x_1x_2} M_y^b\}$
	\\
	Remark that this new representation is obtained as the tensor of the previous one. Hence, it can be already partially decomposed into $T\otimes T\oplus T\otimes S\oplus T\otimes\phi\oplus \pi_+\otimes T\oplus \pi_+\otimes S\oplus \pi_+\otimes \phi\oplus \phi\otimes T\oplus \phi\otimes S\oplus \phi\otimes \phi$. The terms tensor-ed with $T$ are already irreps.
	
	It decomposes as $T^{\otimes 2}\oplus\phi^{\otimes 5}\oplus\pi_+^{\otimes 3}\oplus\pi_-^{\otimes 3}\oplus\Lambda\oplus\Omega\oplus\lambda\oplus\omega$, with:
	\begin{itemize}
		
		\item One $T$, one $\pi_+$, one $S$, two $\phi$ coming from $\alpha\otimes\beta$ where $\alpha$ or $\beta$ is $T$. A basis is obtained by tensorisation of the basis of $\alpha$ and $\beta$.
		
		\item  $S\otimes \pi_+$ is irreducible and called $\pi_-$. A basis is obtained by tensorisation.
		
		\item  $\phi\otimes S$ is isomorphic to $\phi$, with a symmetry adapted basis $\xi_i^{X_1}\otimes\delta_+^{X_2}\otimes\delta_-^{Y}\otimes\delta_+^{B}, -\delta_+^{X_1}\otimes\xi_i^{X_2}\otimes\delta_-^{Y}\otimes\delta_+^{B}$.
		
		\item  $\pi_+\otimes \phi$ decomposes as $\pi_+\otimes  \phi=\phi\oplus\pi_+\oplus\pi_-\oplus\lambda\oplus\omega$. A basis of $\pi_+\otimes \phi$ can be obtained by tensorisation. For simplicity in the notations, we first do the following identification:
		$\xi_i^{X_1}\otimes\xi_j^{X_2}\otimes e_k^{Y}\otimes\xi_l^{B}\cong e_l \xi_i\xi_k\xi_j$, in which we omitted the tensor products for compactness.
		In the following, we group $\xi_i\xi_k$ for $l=1$ and $\xi_k\xi_j$ for $l=2$ to form the representation $\phi_1\otimes\phi_1=\Lambda\phi_1 \oplus t \oplus \phi_1 \oplus \theta$. Hence we obtain basis vectors $\tilde\beta_{i}, \tilde\delta_+, \tilde\xi_i$ and $\tilde u_k$ which are created out of $\xi_i\xi_k$ for $l=1$ and $\xi_k\xi_j$ for $l=2$. Then, we find the following:
		\begin{itemize}
			\item A copy of $\phi$ generated by the $e_1\tilde{\delta}_+\xi_j$, $e_2\xi_i\tilde{\delta}_+$.
			\item A copy of $\pi_+$ generated by the $\propto e_1\tilde\beta_{i}\xi_j+e_2\xi_i\tilde\beta_{j}$.
			\item A copy of $\pi_-$ generated by the $\propto e_1\tilde\beta_{i}\xi_j-e_2\xi_i\tilde\beta_{j}$.
			\item A copy of $\lambda$ generated by the $e_1\tilde{\beta}_i\xi_j$, $e_2\xi_i\tilde{\beta}_j$.
			\item A copy of $\omega$ generated by the $e_1\tilde{u}_i\xi_j$, $e_2\xi_i\tilde{u}_j$.
		\end{itemize}
		
		\item  $\phi\otimes\phi$ decomposes as $T\oplus S\oplus\phi\oplus\pi_+\oplus\pi_-\oplus\Lambda\oplus\Omega$. 
		A basis of $\phi\otimes\phi$ can be obtained by tensorisation. For simplicity in the notations, we first do the following identification:
		$\xi_i^{X_1}\otimes\delta_+^{X_2}\otimes e_k^{Y}\otimes\xi_j^{B}\cong e_1 e_k\xi_i\xi_j$
		and $\delta_+^{X_1}\otimes\xi_i^{X_2}\otimes e_k^{Y}\otimes\xi_j^{B}\cong e_2 e_k\xi_i\xi_j$, in which we omitted the tensor products for compactness.
		Then $\phi\otimes\phi$ contains the following irreps:
		\begin{itemize}
			\item A copy of $\pi_+$ generated by the $e_1e_2\xi_i\xi_j+e_2e_1\xi_j\xi_i$.
			\item A copy of $\pi_-$ generated by the $e_1e_2\xi_i\xi_j-e_2e_1\xi_j\xi_i$.
		\end{itemize}
		Remark that the remaining vectors are of the form $e_ke_k\xi_i\xi_j$: the decomposition $\phi_1\otimes\phi_1=\Lambda\phi_1 \oplus t \oplus \phi_1 \oplus \theta$ now appears. We write $\tilde\beta_{i}, \tilde\delta_+, \tilde\xi_i$ and $\tilde u_k$ the corresponding basis constructed out of this $\xi_i\xi_j$ as explained before.Then, we find:
		\begin{itemize}
			\item A copy of $T$ generated by the $e_1e_1\delta_++e_2e_2\delta_+$.
			\item A copy of $S$ generated by the $e_1e_1\delta_+-e_2e_2\delta_+$.
			\item A copy of $\phi$ generated by the $e_ke_k\tilde{\xi}_i$.
			\item A copy of $\Lambda$ generated by the $e_ke_k\tilde{\beta}_{i}$.
			\item A copy of $\Omega$ generated by the $e_ke_k\tilde{u}_{i}$.
		\end{itemize}

	\end{itemize}
	
	Finally, note that this analytical block decomposition is numerically implemented in the software package.

		\section{Proof of self-test of SIC-POVM}
		In this section, we prove the self-testing result of the main text, i.e., we derive the implications of observing the maximal quantum value ($W^Q_{d}$) of the witness
		\begin{equation}\label{witness}
		W_{d}=\underbrace{\sum_{x<x'}\left[P(b=0|x,(x,x'))+P(b=1|x',(x,x'))\right]}_{\equiv T}+\underbrace{\sum_{x=1}^{N}P(o=x\lvert x,\mathbf{povm})}_{\equiv R}.
		\end{equation} 
		We will show that for $N=d^2$, finding $W_{d}=W_d^Q$ implies that Alice prepares $N$ pure states $\rho_x=\ketbra{\psi_x}{\psi_x}$ such that  
		\begin{equation}\label{rel}
		\lvert \braket{\psi_x}{\psi_{x'}}\lvert ^2 =\frac{1}{d+1}
		\end{equation}
		for $x\neq x'$, and that the setting $\mathbf{povm}$ of Bob corresponds to a SIC-POVM. That is, the measurement can be written as $\{\frac{1}{d}\ketbra{\psi_x}{\psi_x}\}_{x=1}^{d^2}$.


		We begin by focusing on the first sum in \eqref{witness}, and later take the second sum into account. In quantum theory, the maximal value of the first sum in \eqref{witness}  reads
		\begin{align}
		T^Q&\equiv\max_{\{\rho\}, \{M\}} \sum_{x<x'}\left[P(b=0|x,(x,x'))+P(b=1|x',(x,x'))\right]\\
		&= \max_{\{\rho\}, \{M\}} \sum_{x<x'} \Tr\left[(\rho_x-\rho_{x'})M_{(x,x')}^0\right]+\binom{N}{2}= \max_{\{\rho_x\}} \sum_{x<x'}  \lambda_+\left[\rho_x-\rho_{x'}\right]+\binom{N}{2},
		\end{align}
		where we have optimally chosen $M^0_{(x,x')}$ to be the projector onto the positive eigenspace of $\rho_x-\rho_{x'}$, and by $\lambda_+$ denoted the sum of all positive eigenvalues. However, since $W_{d}$ is a linear combination of probabilities obtained over a bounded Hilbert space, the optimal preparations are pure states ($\rho_x=\ketbra{\psi_x}{\psi_x}$). Consequently, for optimal preparations, the operator $\rho_x-\rho_{x'}$ has at most one positive eigenvalue. Hence,
		\begin{equation}
		T^Q=\max_{\{\psi_x\}}  \sum_{x<x'} \lambda_{\text{max}}\left[\ketbra{\psi_x}{\psi_x}-\ketbra{\psi_{x'}}{\psi_{x'}}\right]+\binom{N}{2}.
		\end{equation}
		A pair of states $\ket{\psi_x}$ and $\ket{\psi_{x'}}$ can be effetively parameterised by qubits embedded in a $d$-dimensional Hilbert space. Applying a suitable unitary, we can write two such states as $\ket{\psi}=\ket{0}$ and $\ket{\phi}=\alpha\ket{0}+\beta\ket{1}$ for some complex coefficients $\alpha$ and $\beta$ with $\rvert \alpha\lvert ^2+\rvert \beta\lvert ^2=1$. Solving the characteristic equation $\det\left[\ketbra{\psi}{\psi}-\ketbra{\phi}{\phi}-\lambda\openone\right]=0$, one finds the eigenvalues  $\lambda=\pm |\beta|=\pm \sqrt{1-|\alpha|^2}$. Thus we have
		\begin{equation}
		\lambda_{\text{max}}[\ketbra{\psi}{\psi}-\ketbra{\phi}{\phi}]=\sqrt{1-|\braket{\psi}{\phi}|^2}.
		\end{equation}
		Consequently, 
		\begin{equation}\label{Tq}
		T^Q=\max_{\{\psi_x\}} \sum_{x<x'} \sqrt{1-|\braket{\psi_x}{\psi_{x'}}|^2}+\binom{N}{2}.
		\end{equation}
		We  can now apply the following concavity inequality: for $s_i\geq 0$ and a positive integer $n$, it holds that
		\begin{equation}\label{conc}
		\sum_{i=1}^n \sqrt{s_i}\leq \sqrt{n\sum_{i=1}^{n} s_i},
		\end{equation}
		with equality if and only if all $s_i$ are equal. Applying this to \eqref{Tq} leads to
		\begin{equation}\label{upperbound}
		T^Q\leq \max_{\{\psi_x\}} \sqrt{\binom{N}{2}^2-\binom{N}{2}\sum_{x<x'} |\braket{\psi_x}{\psi_{x'}}|^2}+\binom{N}{2}.
		\end{equation}
		We must now minimise the sum under the square-root. To this end, we write it as 
		\begin{equation}
		\sum_{x<x'} |\braket{\psi_x}{\psi_{x'}}|^2=\frac{1}{2}\left[\sum_{x',x} |\braket{\psi_x}{\psi_{x'}}|^2-N\right].
		\end{equation}
		However, since $\ket{\psi_x}$ is unconstrained other than being of dimension $d$, the sum appearing on the right-hand-side is known as the frame-potential and its known minimum is $N^2/d$ (when $N\geq d$) \cite{framepotential}. Thus we find that
		\begin{equation}\label{boundTq}
		T^Q\leq \sqrt{\frac{N^3\left(N-1\right)\left(d-1\right)}{4d}}+\binom{N}{2}.
		\end{equation}
		Note that this bound on $T^Q$ was first obtained in \cite{BNV13}. 
		
		Let us now focus on the case of interest, namely $N=d^2$. The bound \eqref{boundTq} is tight if and only if we can  ensure equality in our use of the concavity inequality  \eqref{conc} in Eq.~\eqref{upperbound}. Equality is achieved if and only if  $\forall x<x': |\braket{\psi_x}{\psi_{x'}}|^2=c$ for some constant $c$. Using Eq.~\eqref{Tq}, we immediately obtain that $c=1/(d+1)$. Hence, $T=T^Q$ implies that Alice prepares a SIC-ensemble.

		Next, we proceed to include the second sum in the witness \eqref{witness}. We denote the POVM-elements corresponding to the setting $\mathbf{povm}$ by $\{M_{\mathbf{povm}}^o\}_o$. Then, we have that 
		\begin{align}
		R^Q&\equiv \max_{\{\rho\},\{M_{\mathbf{povm}}\}}\sum_{x=1}^{N}P(o=x\lvert x,\mathbf{povm})= \max_{\{\rho_x\},\{E_x\}} \sum_{x=1}^{N} \Tr\left(\rho_x M_{\mathbf{povm}}^x\right)\\
		& \leq \max_{\{M_{\mathbf{povm}}^x\}} \sum_{x=1}^{N}  \lambda_{\text{max}}\left[M_{\mathbf{povm}}^x\right]\leq \max_{\{M_{\mathbf{povm}}^x\}} \Tr\left[\sum_{x=1}^N M_{\mathbf{povm}}^x\right]=d.
		\end{align}
		The first inequality is saturated if and only if $\rho_x$ is a pure state aligned with the eigenvector corresponding to the largest eigenvalue of $M_{\mathbf{povm}}^x$. The second inequality is saturated if and only if $\forall x$ $M_{\mathbf{povm}}^x$ is rank-one. The maximal quantum value of  $W_{d}$ is upper bounded by $T^Q+R^Q$. Since observing $T=T^Q$ implies that Alice's ensemble is SIC, it implies that in order to find $R=R^Q$ one requires $\{M_{\mathbf{povm}}^x\}$ to be rank-one and aligned with the ensemble $\{\ket{\psi_x}\}_{x=1}^N$. This identifies a SIC-POVM. Hence, finding $W_d=T^Q+R^Q$ uniquely implies that $\{M_{\mathbf{povm}}^x\}$ is a SIC-POVM. We conclude that the 
		\begin{equation}
		W^Q_{d}=\frac{1}{2}\sqrt{d^5(d-1)^2(d+1)}+\binom{d^2}{2}+d.
		\end{equation}
		self-tests that Alice prepares a SIC-ensemble and that Bob's setting $\mathbf{povm}$ corresponds to a SIC-POVM.

		\section{Symmetries for certifying non-projective measurements based on SIC-POVMs}
		In this section, we discuss in detail the symmetries of the witness $W_d$ introduced in the main text. The relations between Alice's input and Bob's inputs and outputs that constitute a successful contribution to the value of $W_d$ read
		\begin{align}\nonumber
		& o=x  \hspace{7mm} \text{when Bob has setting } \mathbf{povm}\\\nonumber
		& b=0  \hspace{7mm} \text{when Bob has setting } (y,y') \text{ and Alice has input } x=y\\\label{wincond}
		& b=1  \hspace{7mm} \text{when Bob has setting } (y,y') \text{ and Alice has input } x=y'.
		\end{align}
		
		We first identify the symmetries of the CCP, i.e. the transformations the preserve the winning conditions \eqref{wincond} under general quantum strategies, and then consider a restriction of the symmetries to quantum strategies with projective measurements. Let $S_N$  be the set of $N$-element permutations. We may permute Alice's input with $\omega\in S_N$, i.e. $x\rightarrow \omega(x)$. In order to preserve the winning condition $o=x$ for Bob's setting $\mathbf{povm}$, we therefore need to apply the same permutation to $o$, i.e. $o\rightarrow \omega(o)$. Similar re-labellings apply to  Bob's remaining settings $(y,y')$: the winning conditions $b=0$ when $x=y$, and $b=1$ when $x=y'$, are preserved by letting $(y,y')\rightarrow (\omega(y),\omega(y'))$. However, sometimes we will find that $\omega(y)>\omega(y')$ which does not constitue a proper measurement label. Therefore, whenever this is the case, we swap the labels, i.e. $(\omega(y),\omega(y'))\rightarrow (\omega(y'),\omega(y))$. The swap will preserve the summand of \eqref{witness} if we additionally also let $b\rightarrow b+1\mod{2}$. 
		
		Now, we consider the symmetries of $W_d$ under projective measurements only. Note first that due to their binary outcomes, the settings $(y,y')$ are always optimally implemented as projective measurements (these are extremal). Hence, we must only constrain the setting $\mathbf{povm}$ to be a  projective measurement. This means that at most $d$ of the POVM elements $\{M_\mathbf{povm}^x\}_{x=1}^{d^2}$ are non-zero, corresponding to  rank-one projectors. Without loss of generality we can choose these to correspond to the outcomes $o=1,\ldots,d$.  Evidently, although every $\omega \in S_N$ preserves the witness, not every $\omega$ preserves the projective constraint on the setting $\mathbf{povm}$. Therefore, we define $S_N^d$ as the the set of permutations of $N$ elements that do not affect either the first $d$ objects, or the last $N-d$ objects. This means that the rank-one and the zero projectors associated to $\mathbf{povm}$ will respectively be permuted amongst themselves. Consequently, every $\omega\in S_N^d$ preserves both the witness and the projective constraint. This fully characterises the set of symmetries used in the main text.


\begin{thebibliography}{99}
\bibitem{CommentI3322}
Although our moment matrix is bigger in absolute terms compared to the one reported in~\cite{NV2}, we obtain worse bounds for $d=3$ and $d=4$, when $c=2$; in particular, our bounds for $d=3$ do not match the lower bound given by see-saw optimisations. This is probably due to~\cite{NV2} using a mix of different relaxation levels (which we were not able to reproduce): indeed, a moment matrix of size $184$ does not correspond to NPA levels 2,3,4 or 5 (which have sizes 28, 88, 244 or 628) or local levels 2,3 or 4 (which have sizes 100, 484 or 2116).
\bibitem{footnote4}
The group $G$ is
ambivalent when for any $g \in G$, there exists $h \in G$ such that
$h \circ g \circ h^{- 1} = g^{- 1}$.
\end{thebibliography}
\end{document}